\documentclass[journal]{IEEEtran}
\usepackage{lineno,hyperref}
\usepackage{bm}
\usepackage{mathrsfs}
\usepackage{amsmath}
\usepackage{amssymb}
\usepackage{textcomp}
\usepackage{graphicx}
\usepackage{epstopdf}
\usepackage{subfigure}
\modulolinenumbers[5]
\usepackage{geometry}
\usepackage{color}
\usepackage{amsthm}

\begin{document}


\title{\huge Sparsity-Aware SSAF Algorithm with Individual Weighting Factors for Acoustic Echo Cancellation}

\author{Yi Yu$^{a}$, Tao Yang$^{a}$, Hongyang Chen$^{b}$, Rodrigo C. de Lamare$^{c}$, Yingsong
Li$^{d}$ \\ a) School of Information Engineering, Robot Technology
Used for Special Environment Key Laboratory of Sichuan Province,
Southwest University of Science and Technology, Mianyang, 621010,
China b) the Research Center for Intelligent Networking, Zhejiang
Lab, Hangzhou 311121, China c) CETUC, PUC-Rio, Rio de Janeiro
22451-900, Brazil, and Department of Electronic Engineering,
University of York, York YO10 5DD, U.K. d) School of Information and
Communication Engineering, Harbin Engineering University, Harbin
150001, China}


\maketitle

\begin{abstract}
In this paper, we propose and analyze the sparsity-aware sign subband adaptive filtering with individual weighting factors (S-IWF-SSAF) algorithm, and consider its application in acoustic echo cancellation (AEC). Furthermore, we design a joint optimization scheme of the step-size and the sparsity penalty parameter to enhance the S-IWF-SSAF performance in terms of convergence rate and steady-state error. A theoretical analysis shows that the S-IWF-SSAF algorithm outperforms the previous sign subband adaptive filtering with individual weighting factors (IWF-SSAF) algorithm in sparse scenarios. In particular, compared with the existing analysis on the IWF-SSAF algorithm, the proposed analysis does not require the assumptions of large number of subbands, long adaptive filter, and paraunitary analysis filter bank, and matches well the simulated results. Simulations in both system identification and AEC situations have demonstrated our theoretical analysis and the effectiveness of the proposed algorithms.
\end{abstract}

\begin{IEEEkeywords}
Acoustic echo cancellation; impulsive noise; performance analysis; sign subband adaptive filter; sparse system.
\end{IEEEkeywords}



\section{Introduction}
Adaptive filtering algorithms have been extensively developed in Gaussian noise environments~\cite{sayed2003fundamentals,lee2009subband}, and representative examples are the least mean square (LMS) and normalized LMS (NLMS) algorithms. However, impulsive noise is often encountered in realistic environments such as echo cancellation, underwater acoustics, audio processing, communications, and prediction of time-series~\cite{nikias1995signal,zimmermann2002analysis,georgiou1999alpha,vega2008new,yu2019dcd,LDang2019}. Although the impulsive noise appears randomly with a small probability or a short duration, its realizations have large amplitude. In this situation, the algorithms based on Gaussian noise suffer from a poor convergence or even divergence. Aiming at impulsive noise, Mathews~\textit{et al.} first proposed the sign algorithm which minimizes the absolute value of the error signal~\cite{mathews1987improved}. The maximum correntropy criterion was frequently studied to present efficient and robust LMS-like algorithms in impulsive noise~\cite{chen2016generalized,chen2017kernel,chen2014steady}, thanks mainly to the strong compression capability of the correntropy function on the error signals with large amplitude.

The problem of the aforementioned algorithms is { the slow
convergence when the input signal to the adaptive filter is highly
correlated.} To speed up the convergence, the subband adaptive
filter (SAF) is one of the promising
approaches~\cite{lee2009subband}. In the SAF, the input signal is
divided into multiple subband components through the analysis
filters, { and then the decimated subband input signals that have
approximately uncorrelated samples} are used for updating the
filter's weights. Among the SAF's structures, its multiband
structure updates the fullband-like filter's weights is more
promising due to avoiding the aliasing and band edge
effects~\cite{lee2009subband}. Based on the multiband structure, the
normalized SAF (NSAF) algorithm in~\cite{lee2004improving} converges
faster than the NLMS algorithm for highly correlated inputs. The
excess computational complexity of the NSAF over the NLMS is
trivial, especially for long filter applications such as echo
cancellation. By incorporating the sign algorithm into the SAF, the
sign subband adaptive filter (SSAF) algorithm was proposed
in~\cite{ni2010signsubband}, with good robustness against impulsive
noise. By fully taking advantage of the decorrelation feature of
SAF, the individual-weighting-factors based SSAF (IWF-SSAF)
algorithm~\cite{yu2016novel} provides faster convergence than the
SSAF algorithm. { Note that, when users choose the fixed step-size,
both SSAF algorithms need to consider a trade-off between fast
convergence and low steady-state error. To address this problem},
several variable step-size (VSS) strategies were developed for both
algorithms that drive fast convergence and low steady-state error
simultaneously, to name a few, the VSS-SSAF~\cite{shin2013variable}
and band-dependent VSS SSAF (BDVSS-SSAF)~\cite{yoo2014band}
algorithms from the mean-square deviation (MSD) minimization, and
the novel VSS SSAF~\cite{kim2013sign} and band-dependent VSS
IWF-SSAF (BDVSS-IWF-SSAF)~\cite{kim2017delayless} algorithms based
on the $l_1$-norm minimization, and the robust VSS SSAF
algorithm~\cite{wen2017robust}. Among them, the VSS strategies
in~\cite{kim2013sign} and~\cite{kim2017delayless} do not require the
\emph{a priori} knowledge of the surrounding noise, e.g., the noise
variance and the occurrence probability of impulsive noise. The
performance analysis is always a vital research issue for adaptive
filtering
algorithms~\cite{sayed2003fundamentals,paul2011convergence,zheng2017steady,de2005statistical,Zheng2019Steady-State,zheng2017robust}.
It is beneficial to support the effectiveness of a specific adaptive
filtering algorithm in theory and to give useful insights to further
improve the adaptive filter's performance. { It is worth noting that
one often pays attention to the performance analyses of adaptive
filtering algorithms in Gaussian noise surroundings.} For example,
there have been different analysis models on the MSD behavior of the
NSAF
algorithm~\cite{yin2011stochastic,shin2018adaptive,jeong2016mean,zhang2019mean}.
Relatively speaking, it is more difficult to analyze the performance
of robust adaptive filtering algorithms in impulsive noise.
In~\cite{yu2016steady}, the steady-state MSD of the SSAF algorithm
in impulsive noise was analyzed based on the energy conservation
relation, and in~\cite{yu2016novel} the same analysis pattern was
also extended to the IWF-SSAF algorithm. However, this analysis
approach relies on the assumptions of large number of subbands and
long adaptive filter. By assuming the background noise to be
Gaussian, the analytical expression in~\cite{shin2017steady} shows
better accuracy than the one in~\cite{yu2016steady} for the
steady-state MSD of the SSAF algorithm. { Nevertheless, the
advantage of the SSAF algorithm is working in impulsive noise, so in
this case the analysis in~\cite{shin2017steady} is not applicable.}

{ On the other hand, the aforementioned algorithms have not
exploited the underlying sparsity of the systems.} Sparse systems
are common in practice, with the property that its impulse response
only has a few large non-zero coefficients (\emph{active
coefficients}) and the remaining coefficients are zero or approach
zero (\emph{inactive coefficients}), such as network/acoustic echo
channels~\cite{radecki2002echo,hansler2006topics}, underwater
acoustic channels~\cite{pelekanakis2014adaptive}, and digital TV
transmission channels~\cite{schreiber1995advanced}. In order to
favor such sparsity, the sparsity-aware technique is popular in
adaptive filtering
algorithms~\cite{haddad2014transient,gu2009norm,de2014sparsity,yu2019sparsity}
that adds the sparse constraint term in the original cost function.
{ In the survey of robust SAF against impulsive noise,
sparsity-aware approaches were only incorporated straightforwardly
into the SSAF~~\cite{choi2014new} and normalized logarithmic
SAF~\cite{shen2018L0} algorithms, and have not been analyzed
theoretically yet. Also, the resulting algorithms require properly
choosing the sparsity penalty parameter in a trial and error way,
thereby limiting their usefulness.}

{ It is remarked that the acoustic echo cancellation (AEC)
application involves the above-mentioned three characteristics: high
correlation of speech input signals, double-talk that is one type of
impulsive noise scenarios, and sparsity of acoustic echo channels.
For the sake of these requirements, therefore, this paper will focus
on studying} the sparsity-aware IWF-SSAF (S-IWF-SSAF) algorithm. Our
main contributions are as follows:

1) By incorporating the sparsity-aware technique, we propose the S-IWF-SSAF algorithm, and analyze its performance in-depth in impulsive noise. The analysis result reveals that S-IWF-SSAF can be superior to IWF-SSAF in sparse system environments, but it requires properly choosing the sparsity penalty parameter in a certain range.

2) The proposed analysis covers the behaviors of the IWF-SSAF
algorithm in impulsive noise. { Even though for the IWF-SSAF
algorithm, the proposed analysis} is significantly more accurate
than the analysis in~\cite{yu2016novel} and closely matches the
simulations, because it obviates the assumptions of a large number
of subbands, long adaptive filter, and paraunitary analysis filter
bank.

3) We devise a joint optimization scheme to automatically choose the step-size and the sparsity penalty parameter, which further improves the convergence and steady-state performance of the S-IWF-SSAF algorithm.

4) { In order to make the proposed algorithms suitable for AEC}, we
develop delayless implementation of the proposed algorithms and
carry out a simulation study.

The notations used in this paper are listed in Table~1. This paper is organized as follows. In Section~2, we state the SAF problem of interest and briefly review the IWF-SSAF algorithm. Then, we propose the S-IWF-SSAF algorithm in Section~3 and analyze its performance in Section~4. In Section~5, the VSS mechanism and the variable sparsity penalty parameter for the S-IWF-SSAF algorithm are devised. In Section~6, simulation results in both system identification and AEC scenarios are presented. Finally, conclusions are given in Section~7.
\begin{table}[tbp]
    \scriptsize
    \centering
    \vspace{-1em}
    \caption{Some Mathematical Symbols.}
    \label{table_3}
    \begin{tabular}{l|c}
        \hline
        \textbf{Notations} &\textbf{Description}\\
        \hline
        $(\cdot)^\text{T}$     &\text{transpose of a vector or matrix}\\
        $||\cdot||_2$   &\text{$l_2$-norm of a vector}\\
        $\text{sgn}(\cdot)$  &\text{sign function}\\
        $\text{E}\{\cdot\}$  &\text{expectation of a random variable} \\
        $\text{Tr}(\cdot)$  &\text{trace of a matrix}\\
        $\text{vec}(\cdot)$  &yielding an $L^2\times 1$ vector from an $L\times L$ matrix by successively stacking the columns of the matrix\\
        $\text{vec}^{-1}(\cdot)$  &the inverse operation of $\text{vec}(\cdot)$\\
        $\otimes$ &Kronecker product of two matrices\\
        \hline
    \end{tabular}
\end{table}

\section{Problem Statement and The IWF-SSAF Algorithm}
Let us consider a system identification problem that identifies the impulse response of the unknown system, denoted as an $M$-length column vector. By feeding the input signal $u(n)$ into the unknown system, the desired signal $d(n)$ of the system at discrete time~$n$ is formulated as
\begin{align}
\label{001}
d(n) = \bm u^\text{T}(n) \bm w^o+v(n)
\end{align}
where $\bm u(n)=[u(n), u(n-1),...,u(n-M+1)]^\text{T}$ is the $M\times1$ input vector, and $v(n)$ denotes the additive noise.
\begin{figure}[htb]
    \centering
    \includegraphics[scale=0.35]{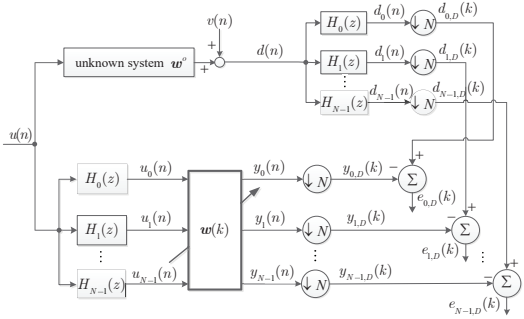}
    \vspace{-1em} \caption{Multiband structure of SAF.}
    \label{Fig1}
\end{figure}

Fig.~\ref{Fig1} shows the multiband structure of SAF with~$N$ subbands~\cite{lee2009subband}. The desired signal $d(n)$ and the input signal $u(n)$ are decomposed into multiple subband signals $d_i(n)$ and $u_i(n)$, respectively, through the analysis filter bank $\{H_i(z)\}_{i=0}^{N-1}$. Then, for each band~$i$, the signal $u_i(n)$ is filtered by the common filter (whose weight vector is $\bm w(k)$) to get the output signal $y_i(n)$. By critically decimating sequences $d_i(n)$ and $y_i(n)$, respectively we obtain lower sampled-rate sequences $d_{i,D}(k)$ and $y_{i,D}(k)$, namely, $d_{i,D}(k)=d_i(kN)$ and $y_{i,D}(k)=\bm u^\text{T}_i(k) \bm w(k)$, where $\bm u_i(k)=[u_i(kN),u_i(kN),...,u_i(kN-M+1)]^\text{T}$. By subtracting $y_{i,D}(k)$ from $d_{i,D}(k)$ for $i=0,...,N-1$, the decimated subband error signals used for updating $\bm w(k)$ are obtained:
\begin{equation}
\label{002}
\begin{array}{rcl}
\begin{aligned}
e_{i,D}(k) = d_{i,D}(k) - \bm u^\text{T}_i(k) \bm w(k).
\end{aligned}
\end{array}
\end{equation}
Since $\bm w(k)$ is an estimate of~$\bm w^o$ in the decimated domain~$k$, the main task of the adaptive filter is how to adjust $\bm w(k)$ by using the available signals $\{d_{i,D}(k),u_i(k),e_{i,D}(k)\}_{i=0}^{N-1}$ to reach quickly $\bm w(k)\rightarrow \bm w^o$ as $k$ increases.

To estimate $\bm w^o$ in impulsive noise, the cost function with the band-dependent weighting factor $\lambda_i$ is defined as
\begin{equation}
\label{003}
\begin{aligned}
J(k)=\sum \limits_{i=0}^{N-1}\zeta_i |e_{i,D}(k)|.
\end{aligned}
\end{equation}
By applying the instantaneous gradient descent (IGD) principle to minimize~\eqref{003} with respect to $\bm w(k)$ and letting $\zeta_i=1/\|\bm u_i(k)\|_2$, the weight vector of the IWF-SSAF algorithm is recursively updated~\cite{yu2016novel}:
\begin{equation}
\label{004}
\begin{aligned}
\bm w(k+1)=\bm w(k) + \mu \sum \limits_{i=0}^{N-1} \frac{\text{sgn}(e_{i,D}(k))\bm u_i(k)}{\|\bm u_i(k)\|_2}
\end{aligned}
\end{equation}
where $\mu>0$ is the step-size. Due to the sign function $\text{sgn}(\cdot)$, the IWF-SSAF algorithm is robust against impulsive noise.

\section{Proposed S-IWF-SSAF algorithm}
Based on sparsity-aware techniques~\cite{haddad2014transient,gu2009norm,de2014sparsity,yu2019sparsity,tsp2020yu}, we modify the cost function in~\eqref{003} as
\begin{equation}
\label{0b36}
\begin{aligned}
J_s(k)=\sum \limits_{i=0}^{N-1}\zeta_i |e_{i,D}(k)| + \rho H(\bm w),
\end{aligned}
\end{equation}
where $H(\bm w)$ is the penalty term for favoring the sparsity of $\bm w^o$, and $\rho>0$ is the sparsity penalty parameter associated with this term. By applying the IGD principle to reach a minimization of~\eqref{0b36}, the proposed S-IWF-SSAF algorithm for updating the weight vector is formulated as
\begin{equation}
\label{0b37}
\begin{aligned}
\bm w(k+1)=\bm w(k) + \mu \sum \limits_{i=0}^{N-1} \frac{\text{sgn}(e_{i,D}(k))\bm u_i(k)}{\|\bm u_i(k)\|_2} - \rho H'(\bm w(k)),
\end{aligned}
\end{equation}
where $H'(\bm w) \triangleq \frac{\partial H(\bm w) }{\partial \bm w}$ is the subgradient of $H'(\bm w)$ with respect to $\bm w$, and here the parameter $\rho$ has absorbed $\mu$. Furthermore, \eqref{0b37} can be implemented equivalently by two steps:
\begin{subequations} \label{eq:b38}
    \begin{align}
    \bm \varphi(k+1) &= \bm w(k) + \mu \sum \limits_{i=0}^{N-1} \frac{\text{sgn}(e_{i,D}(k))\bm u_i(k)}{\|\bm u_i(k)\|_2},
    \label{eq:b38a}\\
    \bm w(k+1)&= \bm \varphi(k+1)- \rho H'(\bm \varphi(k+1)).\label{eq:b38b}
    \end{align}
\end{subequations}
In equation~\eqref{eq:b38}, the first step plays the adaptive learning role of the IWF-SSAF algorithm, and the second step drives inactive coefficients in $\bm \varphi(k+1)$ to zero to further improve the estimation performance of $\bm w^o$. Note that, in~\eqref{eq:b38b} we have also made a replacement of $H'(\bm w(k))$ with $H'(\bm \varphi(k+1))$. By doing so, we are able to relatively independently design the adaptive choices of $\mu$ and $\rho$ according to~\eqref{eq:b38a} and~\eqref{eq:b38b}, respectively, which will be discussed in the sequel. Now, another issue is the choice of the penalty term~$H(\bm w)$. In the literature, several strategies on $H(\bm w)$ have been reported to present different sparsity-aware adaptive filtering algorithms, see~\cite{tsp2020yu,gu2009norm,de2014sparsity} and references therein. Here, we do not discuss the effect of these strategies, and choose the popular log-penalty:
\begin{equation}
\label{0b39}
\begin{aligned}
H(\bm \varphi(k+1)) = \sum_{m=1}^{M} \ln(1+|\varphi_m(k+1)|/\xi),
\end{aligned}
\end{equation}
where $\varphi_m(k+1)$ denotes the $m$-th entry of $\bm \varphi(k+1)$, and $\xi>0$ denotes the shrinkage magnitude that distinguish active entries and inactive entries. Accordingly, $H'(\bm \varphi(k+1))$ in~\eqref{eq:b38b} is given by
\begin{equation}
\label{040}
\begin{aligned}
H'(\varphi_m(k+1)) = \frac{\text{sgn}(\varphi_m(k+1))}{\xi + |\varphi_m(k+1)|}, m=1,...,M.
\end{aligned}
\end{equation}
We note that the proposed algorithms can also be considered in the
context of detection problems
\cite{spa,mfsic,gmibd,tdr,mbdf,mmimo,lsmimo,armo,siprec,did,lclrbd,gbd,wlbd,bbprec,mbthp,rmbthp,badstbc,bfidd,1bitidd,1bitce,baplnc,jpbnet,mwc,aaidd,listmtc}
and might be further enhanced by exploitation of low-rank techniques
\cite{intadap,inttvt,jio,ccmmwf,wlmwf,jidf,jidfecho,barc,jiols,jiomimo,jiostap,sjidf,l1stap,saabf,jioccm,ccmrab,wlbeam,lcrab,jiodoa,rrser,rcb,saalt,dce,damdc,locsme,memd,okspme,lrcc,rrdoa,kaesprit,rhomo,sorsvd,corutv,sparsestap,stapcoprime,wlccm,dcdrec,kacs}.

\section{Performance Analysis}
In this section, we study in detail the statistical performance of the S-IWF-SSAF algorithm in the presence of impulsive noise, and in theory illustrate its superiority over the IWF-SSAF algorithm when identifying sparse systems.

Assuming that $\bm w^o$ is time-invariant. Then, from~\eqref{eq:b38a} and~\eqref{eq:b38b} we obtain the following weights error vector recursions:
\begin{subequations} \label{eq:041}
    \begin{align}
    \widetilde{\bm \varphi}(k+1) &= \widetilde{\bm w}(k) - \mu \sum \limits_{i=0}^{N-1} \frac{\text{sgn}(e_{i,D}(k))\bm u_i(k)}{\|\bm u_i(k)\|_2},
    \label{eq:041a}\\
    \widetilde{\bm w}(k+1)&= \widetilde{\bm \varphi}(k+1) + \rho H'(\bm \varphi(k+1)),\label{eq:041b}
    \end{align}
\end{subequations}
where $\widetilde{\bm w}(k) \triangleq \bm w^o - \bm w(k)$ and $\widetilde{\bm \varphi}(k) \triangleq \bm w^o - \bm \varphi(k)$.

Let $\{\bm h_i\}_{i=0}^{N-1}$ be the impulse response of the analysis filter bank $\{H_i(z)\}_{i=0}^{N-1}$, with the length of $L$, then the following relations hold:
\begin{equation}
\label{006}
\begin{array}{rcl}
\begin{aligned}
\bm u_i(k) & =[\bm u(kN),...,\bm u(kN-L+1)] \bm h_i\\
v_{i,D}(k)&=\bm h_i^\text{T} [v(kN),...,v(kN-L+1)]^\text{T}.
\end{aligned}
\end{array}
\end{equation}
Therefore, \eqref{002} can be rearranged as
\begin{equation}
\label{007}
\begin{array}{rcl}
\begin{aligned}
e_{i,D}(k) = e_{i,a}(k) + v_{i,D}(k),
\end{aligned}
\end{array}
\end{equation}
where $e_{i,a}(k) \triangleq \bm u^\text{T}_i(k) \widetilde{\bm w}(k)$ denotes the \textit{a priori} subband error. Equations~\eqref{eq:041}-\eqref{007} will be the starting point to study the mean and mean-square behaviors of the algorithm in the sequel. Moreover, in order to help to analyze the algorithm mathematically, some widely used assumptions are made as follows:

\textit{Assumption 1}: The input vector $\bm u(n)$ is random with zero-mean vector and positive definite autocorrelation matrix $\bm R_u =\text{E}\{\bm u(n)\bm u^\text{T}(n)\}$.

\textit{Assumption 2}: The additive noise $v(n)$ is drawn from the contaminated-Gaussian (CG) process, i.e., $v(n)=v_\text{g}(n) + v_\text{im}(n)$. Specifically, the background noise $v_\text{g}(n)$ is white Gaussian with zero-mean and variance $\sigma_\text{g}^2$. The impulsive noise component $v_\text{im}(n)$ is described by the Bernoulli-Gaussian model $v_\text{im}(n)=b(n) \eta(n)$, where $b(n)$ obeys the Bernoulli distribution that the probability of occurring 1 is $P\{b(n)=1\}=p_r$, and $\eta(n)$ is also zero-mean white Gaussian but with variance $\sigma_\eta^2 =\hbar \sigma_\text{g}^2$, $\hbar \gg 1$. Thus, it is seen that $v(n)$ is non-Gaussian for any $p_r$ excluding two special cases $p_r=0$ and~1~\cite{bershad2008error}

\textit{Assumption 3}: $\widetilde{\bm w}(k)$ is independent of $\bm u_i(k)$ for $i=0,...,N-1$.

Note that, with the relation~\eqref{006}, assumption~1 shows that the $i$-th subband's input vector $\bm u_i(k)$ is also zero-mean and has a positive definite autocorrelation matrix $\bm R_i = \text{E}\{\bm u_i(k)\bm u_i^\text{T}(k)\}$. Assumption 2 is a popular model for analyzing the performance of adaptive filtering algorithms in impulsive noise~\cite{tsp2020yu,bershad2008error}. Assumption 3 is the well-known~\textit{independence assumption} for the performance analysis of adaptive filtering algorithms~\cite{sayed2003fundamentals,yang2019convergence}.

\subsection{Mean behavior}
By imposing the expectations on both sides of~\eqref{eq:041a} and~\eqref{eq:041b}, we have
\begin{subequations} \label{eq:042}
    \begin{align}
    \text{E}\{\widetilde{\bm \varphi}(k+1)\} &= \left( \bm I_M- \mu \sum \limits_{i=0}^{N-1} \Omega_i(k) \text{E}\{\bm A_i(k)\}\right)  \text{E}\{\widetilde{\bm w}(k)\},
    \label{eq:042a}\\
    \text{E}\{\widetilde{\bm w}(k+1)\}&= \text{E}\{\widetilde{\bm \varphi}(k+1)\} + \rho \text{E}\{H'(\bm \varphi(k+1))\}, \label{eq:042b}
    \end{align}
\end{subequations}
where
\begin{equation}
\label{014}
\begin{aligned}
\Omega_i(k) = \sqrt{\frac{2}{\pi}} \left[ \frac{p_r}{\sqrt{\text{E} \{e_{i,D,1}^2(k)\}}} + \frac{1-p_r}{\sqrt{\text{E} \{e_{i,D,2}^2(k)\}}} \right] \neq 0,
\end{aligned}
\end{equation}
\begin{equation}
\label{015}
\begin{aligned}
\text{E} \{e_{i,D,1}^2(k)\} = \text{E} \{e_{i,a}^2(k)\}+ ||\bm h_i||_2^2 (\hbar + 1)\sigma_\text{g}^2,
\end{aligned}
\end{equation}
and
\begin{equation}
\label{016}
\begin{aligned}
\text{E} \{e_{i,D,2}^2(k)\} = \text{E} \{e_{i,a}^2(k)\}+ ||\bm h_i||_2^2 \sigma_\text{g}^2.
\end{aligned}
\end{equation}
Note that, the detailed process of~\eqref{eq:042a} is shown in Appendix~A. Supposing that the algorithm converges, it holds that~$\text{E}\{\widetilde{\bm w}(k+1)\} = \text{E}\{\widetilde{\bm w}(k)\}$ as $k\rightarrow \infty$. Hence, from~\eqref{eq:042a} and~\eqref{eq:042a} we can deduce the following theorem.

\newtheorem{theorem}{Theorem}
\begin{theorem}
    When $k \rightarrow \infty$, it is established that
    \begin{equation}
    \label{044}
    \begin{aligned}
    \text{E}\{\widetilde{\bm w}(\infty)\}&= \frac{\rho}{\mu} \left(\sum \limits_{i=0}^{N-1} \Omega_i(\infty) \text{E}\{\bm A_i(\infty)\}\right)^{-1} \text{E}\{H'(\bm w(\infty))\},
    \end{aligned}
    \end{equation}
    which points out that the S-IWF-SSAF algorithm is biased for estimating a sparse vector $\bm w^o$.
\end{theorem}

\textit{Remark~1:} For a special case of~$\rho=0$, \eqref{044} leads to
\begin{equation}
\label{017}
\begin{aligned}
\text{E}\{\widetilde{\bm w}(\infty)\} = \bm 0,
\end{aligned}
\end{equation}
that is, the IWF-SSAF algorithm is unbiased for estimating $\bm w^o$ in the presence of impulsive noise.

It is worth noting that the mean transient behavior of both IWF-SSAF and S-IWF-SSAF algorithms relies on its mean-square behavior as we shall discuss below.

\subsection{Mean-square behavior}
Let us define the autocorrelation matrices of the weight error vectors $\widetilde{\bm w}(k)$ and $\widetilde{\bm \varphi}(k)$, as follows, $\widetilde{\bm W}(k) \triangleq \text{E}\{\widetilde{\bm w}(k)\widetilde{\bm w}^\text{T}(k)\}$ and $\widetilde{\bm \varPhi}(k) \triangleq \text{E}\{\widetilde{\bm \varphi}(k) \widetilde{\bm \varphi}^\text{T}(k)\}$. Thus, by equating the autocorrelation matrices for both sides of~\eqref{eq:041a} and~\eqref{eq:041b} respectively, we have the following recursions:
\begin{subequations} \label{eq:043}
    \begin{equation}
    \begin{array}{rcl}
    \begin{aligned}
    \label{eq:043a}
    \widetilde{\bm \varPhi}(k+1) =& \widetilde{\bm W}(k) + \mu^2 \sum \limits_{i=0}^{N-1} \text{E}\{\check{\bm A}_i(k)\}- \\
    &\mu \widetilde{\bm W}(k) \sum \limits_{i=0}^{N-1} \Omega_i(k) \text{E}\{\bm A_i(k)\} - \\
    &\mu \sum \limits_{i=0}^{N-1} \Omega_i(k) \text{E}\{\bm A_i(k)\} \widetilde{\bm W}(k),
    \end{aligned}
    \end{array}
    \end{equation}
    \begin{equation}
    \begin{array}{rcl}
    \begin{aligned}
    \label{eq:043b}
    \widetilde{\bm W}(k+1) =& \widetilde{\bm \varPhi}(k+1) + \rho \bm \varTheta(k+1) +  \\
    &\rho \bm \varTheta^\text{T}(k+1) + \rho^2 \bm \varXi(k+1),
    \end{aligned}
    \end{array}
    \end{equation}\\
\end{subequations}
where $\bm \varTheta(k+1) \triangleq \text{E}\{H'(\bm \varphi(k+1)) \widetilde{\bm \varphi}^\text{T}(k+1)\}$ and $\bm \varXi(k+1) \triangleq \text{E}\{H'(\bm \varphi(k+1)) {H'}^\text{T}(\bm \varphi(k+1))\}$. Note that, the derivation of~\eqref{eq:043a} is given in~Appendix~B. Obviously, the mean model in~\eqref{eq:042b} and the mean-square model in~\eqref{eq:043b} require computing $\text{E}\{H'(\bm \varphi(k+1))\}$, $\bm \varTheta(k+1)$, and $\bm \varXi(k+1)$ beforehand in a component-wise way, as shown in Appendix~C. Furthermore, to implement the recursion in~\eqref{eq:043a} conveniently, $\Omega_i(k)$ given by~\eqref{014} is rewritten as
\begin{equation}
\label{022}
\begin{aligned}
\Omega_i(k) =& \sqrt{\frac{2}{\pi}} \left[ \frac{p_r}{\sqrt{\text{Tr} \{ \widetilde{\bm W}(k)\bm R_i\} + ||\bm h_i||_2^2 (\hbar + 1)\sigma_\text{g}^2 }} + \right. \\
&\left. \frac{1-p_r}{\sqrt{\text{Tr} \{ \widetilde{\bm W}(k)\bm R_i\} + ||\bm h_i||_2^2 \sigma_\text{g}^2 }} \right].
\end{aligned}
\end{equation}

The MSD is defined as $\text{MSD}(k)\triangleq \text{E}\{||\widetilde{\bm w}(k)||_2^2\} = \text{Tr}\{\widetilde{\bm W}(k)\}$~\cite{sayed2003fundamentals}. Consequently, the model in~\eqref{eq:043} describes the transient MSD behavior of the S-IWF-SSAF algorithm in impulsive noise.

To continue the steady-state analysis, we impose the vectorization operation on both sides of~\eqref{eq:043a} and use the Kronecker property that $\text{vec}(\bm X\bm \Sigma \bm Y)=(\bm Y^\text{T} \otimes \bm X) \text{vec}(\bm \Sigma)$ for matrices $\bm X$, $\bm Y$, and $\bm \Sigma$ of compatible sizes\cite{graham2018kronecker}, the following recursion is established:
\begin{equation}
\label{023}
\begin{aligned}
\text{vec}(\widetilde{\bm \varPhi}(k+1)) = \bm F_k \text{vec}(\widetilde{\bm W}(k)) + \mu^2 \text{vec}\left( \sum \limits_{i=0}^{N-1} \text{E}\{\check{\bm A}_i(k)\}\right),
\end{aligned}
\end{equation}
where
\begin{equation}
\label{024}
\begin{aligned}
\bm F_k =& \bm I_{M^2} - \mu \left(\sum \limits_{i=0}^{N-1} \Omega_i(k) \text{E}\{\bm A_i(k)\} \otimes \bm I_M \right)- \\
&\mu \left(\bm I_M \otimes \sum \limits_{i=0}^{N-1} \Omega_i(k) \text{E}\{\bm A_i(k)\} \right).
\end{aligned}
\end{equation}

At the steady-state, it holds that ${\widetilde{\bm W}(k+1)} = \widetilde{\bm W}(k)$, $k \rightarrow \infty$. Therefore, by assuming the existence of $(\bm I_{M^2} - \bm F_\infty)^{-1}$ and applying the operators~$\text{vec}(\cdot)^{-1}$ and $\text{Tr}(\bm X \bm Y)= \text{vec}^\text{T}(\bm X^\text{T}) \text{vec}(\bm Y)$, the steady-state MSD of the S-IWF-SSAF algorithm can be derived from~\eqref{eq:043b} and~\eqref{023} that
\begin{equation}
\label{045}
\begin{array}{rcl}
\begin{aligned}
&\text{MSD}_s(\infty) = \text{MSD}(\infty) +\Delta_s(\infty),
\end{aligned}
\end{array}
\end{equation}
where
\begin{equation}
\label{046}
\begin{array}{rcl}
\begin{aligned}
&\text{MSD}(\infty) =\\
&\;\mu^2 \text{vec}^\text{T}(\bm I_M) (\bm I_{M^2} - \bm F_\infty)^{-1} \text{vec}\left( \sum \limits_{i=0}^{N-1} \text{E}\{\check{\bm A}_i(k)\}\right)
\end{aligned}
\end{array}
\end{equation}
is the result of the IWF-SSAF update~\eqref{eq:b38a} and
\begin{equation}
\label{046x}
\begin{array}{rcl}
\begin{aligned}
\Delta_s(\infty) = &\text{vec}^\text{T}(\bm I_M) (\bm I_{M^2} - \bm F_\infty)^{-1} \times \\
&\text{vec} \left( \rho \bm \varTheta(\infty) + \rho \bm \varTheta^\text{T}(\infty) + \rho^2 \bm \varXi(\infty) \right)
\end{aligned}
\end{array}
\end{equation}
is the result of the sparsity-aware step~\eqref{eq:b38b}.

It is stressed that~\eqref{045} is not a closed-form, as it is self-contained through $\Omega_i(\infty)$. Thus, in terms of $\text{MSD}_s(\infty)$, we may take advantage of some numerical approaches to solve~\eqref{045} or take the MSD value by running~\eqref{eq:043} to the steady-state. From~\eqref{045}, the following theorem can be obtained.
\begin{theorem}
In sparse system scenarios, the steady-state performance of the S-IWF-SSAF algorithm would be superior to that of the IWF-SSAF algorithm, if and only if $\Delta_s(\infty)<0$. Interestingly, the possibility of $\Delta_s(\infty)<0$ is true by choosing~$\rho$ in a range~$0<\rho<\rho_\text{up}$.
\end{theorem}
\begin{proof}
    See Appendix~D.
\end{proof}

\textit{Remark~2 (on the special IWF-SSAF algorithm):} For evaluating the steady-state of the IWF-SSAF algorithm using a small step-size, i.e., \eqref{046}, it can be assumed that~$\text{Tr}(\widetilde{\bm W}(\infty) \bm R_i) \ll ||\bm h_i||_2^2 \sigma_\text{g}^2$. As such, we can make the following approximation for~$\Omega_i(\infty)$:
\begin{equation}
\label{028}
\begin{aligned}
\Omega_i(\infty) \approx \sqrt{\frac{2}{\pi}} \left[ \frac{p_r}{\sqrt{||\bm h_i||_2^2 (\hbar + 1)\sigma_\text{g}^2}} + \frac{1-p_r}{\sqrt{||\bm h_i||_2^2 \sigma_\text{g}^2 }} \right],
\end{aligned}
\end{equation}
which contributes a closed-form of $\text{MSD}(\infty)$. On the other hand, when the number of subbands is large enough, the decimated input signal at each subband can be approximately white~\cite{lee2009subband}. Also, the length of the adaptive filter is required to be long. In this case, it is known from~\cite{yu2016novel} that $\text{E}\{\bm A_i(k)\} \approx \frac{\sigma_{u,i}}{\sqrt{M}}\bm I_M$ and $\text{E}\{\check{\bm A}_i(k)\} \approx \frac{1}{M}\bm I_M$, where $\sigma_{u,i}^2$ is the power of the decimated input signal of the $i$-th subband. Plugging them into~\eqref{046} produces
\begin{equation}
\label{028after}
\begin{aligned}
\text{MSD}_\text{white}(\infty) = \frac{\mu N \sqrt{M}}{2\sum \limits_{i=0}^{N-1} \Omega_i(\infty) \sigma_{u,i}}.
\end{aligned}
\end{equation}
It is seen from \eqref{028after} that the steady-state MSD of the IWF-SSAF algorithm depends on the step size $\mu$, the number of subbands $N$, the filter's length $M$, the background noise variance $\sigma_\text{g}^2$, and the subband input power $\sigma_{u,i}^2$. Specifically, the steady-state MSD becomes large as $\mu$ and $N$ increase. Conversely, increasing $\mu$ and $N$ brings also fast convergence of the algorithm. Hence, this would motivate the improvements of the algorithm in performance by optimizing~$\mu$ and/or~$N$.

\subsection{Stability condition}
Since the elements of $\text{E}\{H'(\bm \varphi(k+1))\}$, $\bm \varTheta(k+1)$, and $\bm \varXi(k+1)$ are bounded, the stability condition of the S-IWF-SSAF algorithm is the same as that of the IWF-SSAF algorithm. As a result, we study the stability condition according to~\eqref{eq:043a}, and then take the traces of all the terms to yield
\begin{equation}
\label{033}
\begin{aligned}
\text{MSD}(k+1) =& \text{MSD}(k) - \Delta(k)
\end{aligned}
\end{equation}
where
\begin{equation}
\label{034}
\begin{aligned}
\Delta(k) =& \mu  \sum \limits_{i=0}^{N-1} \Omega_i(k) \text{Tr}\left( \widetilde{\bm W}(k) \text{E}\{\bm A_i(k)\} \right)  \\
&+\mu \sum \limits_{i=0}^{N-1} \Omega_i(k) \text{Tr}\left(  \text{E}\{\bm A_i(k)\} \widetilde{\bm W}(k) \right)  \\
&-\mu^2 N.
\end{aligned}
\end{equation}

Recalling large $M$ and $N$ assumptions so that $\text{E}\{\bm A_i(k)\}\approx \frac{\sigma_{u,i}}{\sqrt{M}} \bm I_M $ and $\bm R_i \approx \sigma_{u,i}^2 \bm I_M $, thereby $\Delta(k)$ becomes
\begin{equation}
\label{035}
\begin{aligned}
\Delta(k) =2\mu  \sum \limits_{i=0}^{N-1} \Omega_i(k) \frac{\sigma_{u,i}}{\sqrt{M}} \text{MSD}(k) - \mu^2 N,
\end{aligned}
\end{equation}
where
\begin{equation}
\label{036}
\begin{aligned}
\Omega_i(k) =& \sqrt{\frac{2}{\pi}} \left[ \frac{p_r}{\sqrt{\sigma_{u,i}^2 \text{MSD}(k) + ||\bm h_i||_2^2 (\hbar + 1)\sigma_\text{g}^2 }} + \right. \\
&\left. \frac{1-p_r}{\sqrt{\sigma_{u,i}^2 \text{MSD}(k) + ||\bm h_i||_2^2 \sigma_\text{g}^2 }} \right].
\end{aligned}
\end{equation}

Equation~\eqref{033} illustrates that the algorithm converges in the mean-square sense, if and only if $\Delta(k)>0$, which further leads to
\begin{equation}
\label{037}
\begin{aligned}
0 <\mu < \frac{2}{N} \sum \limits_{i=0}^{N-1} \Omega_i(k) \frac{\sigma_{u,i}}{\sqrt{M}} \text{MSD}(k),
\end{aligned}
\end{equation}

It is known from~\eqref{036} that $\Omega_i(k)$ gradually increases during the convergence of the algorithm, that is, its value is the minimum value at the initial iteration~$k=0$. Accordingly, as~$k$ increases, the upper bound of~\eqref{037} increases. The convergence condition of the algorithm is developed:
\begin{equation}
\label{038}
\begin{aligned}
0 <\mu < \frac{2}{N} \sum \limits_{i=0}^{N-1} \Omega_{i,\min} \frac{\sigma_{u,i}}{\sqrt{M}} ||\bm w^o||_2^2,
\end{aligned}
\end{equation}
where
\begin{equation}
\label{039}
\begin{aligned}
\Omega_{i,\min} \triangleq & \Omega_i(0) \\
=& \sqrt{\frac{2}{\pi}} \left[ \frac{p_r}{\sqrt{\sigma_{u,i}^2 ||\bm w^o||_2^2 + ||\bm h_i||_2^2 (\hbar + 1)\sigma_\text{g}^2 }} + \right. \\
&\left. \frac{1-p_r}{\sqrt{\sigma_{u,i}^2 ||\bm w^o||_2^2 + ||\bm h_i||_2^2 \sigma_\text{g}^2 }} \right]
\end{aligned}
\end{equation}
due to $\bm w(0)=\bm 0$. Because $p_r\ll 1$ and $\hbar \gg 1$ in most cases, the first term at the right side of~\eqref{039} is negligible as compared to the second one. Moreover, the roles of $||\bm w^o||_2^2$ are opposite roughly in~\eqref{038} and $\Omega_{i,\min}$ so that the effect of $||\bm w^o||_2^2$ on the upper bound of~\eqref{038} is small when $||\bm w^o||_2^2$ is not far away~1, which can be seen in Fig.~\ref{Fig2}. Based on the above reasons, we can obtain an effective range of values for choosing~$\mu$ in the algorithm:
\begin{equation}
\label{040x1}
\begin{aligned}
0 <\mu < \frac{2}{N} \sum \limits_{i=0}^{N-1} \Omega_{i,\min} \frac{\sigma_{u,i}}{\sqrt{M}},
\end{aligned}
\end{equation}
where
\begin{equation}
\label{040x2}
\begin{aligned}
\Omega_{i,\min} &\stackrel{(a)}{\approx}\sqrt{\frac{2}{\pi}} \frac{1-p_r}{\sqrt{\sigma_{u,i}^2 + ||\bm h_i||_2^2 \sigma_\text{g}^2 }}\\
&\stackrel{(b)}{\approx} \sqrt{\frac{2}{\pi}} \frac{1}{\sqrt{\sigma_{u,i}^2 + ||\bm h_i||_2^2 \sigma_\text{g}^2 }}.
\end{aligned}
\end{equation}
The approximation~$(b)$ in~\eqref{040x2} corresponds the case of no impulsive noise. Although the approximation~$(a)$ is more accurate than the approximation~$(b)$, the latter is more practical on guiding the choice of the step size as it does not require \emph{a priori} information of $p_r$ when $p_r\ll 1$.

c\section{Variable parameters improvements for the S-IWF-SSAF algorithm}
As stated in theorem~2 and remark~2, the S-IWF-SSAF algorithm can be further improved by jointly developing VSS mechanism and adaptation of $\rho$. In this section, we will arrive at this goal.
\subsection{VSS mechanism}
{ By using time-varying and band-dependent step-sizes~$\mu_i(k)$,
$i=0,...,N-1$ to replace $\mu$ in~\eqref{eq:041a}, it} yields
\begin{equation}
\label{047}
\begin{array}{rcl}
\begin{aligned}
\widetilde{\bm \varphi}(k+1) = \widetilde{\bm w}(k) - \sum \limits_{i=0}^{N-1} \mu_i(k) \frac{\text{sgn}(e_{i,D}(k))\bm u_i(k)}{\|\bm u_i(k)\|_2}.
\end{aligned}
\end{array}
\end{equation}
By defining the $i$-th subband intermediate~\emph{a posteriori} error,
\begin{equation}
\label{048}
\begin{array}{rcl}
\begin{aligned}
e_{p,i}(k) = d_{i,D}(k) - \bm u_i^{T}(k) \widetilde{\bm \varphi}(k+1),
\end{aligned}
\end{array}
\end{equation}
and then from~\eqref{047} we are capable of obtaining
\begin{equation}
\label{049}
\begin{array}{rcl}
\begin{aligned}
e_{p,i}(k) = e_{i,D}(k) - \mu_i(k) ||\bm u_i(k)||_2 \text{sgn}(e_{i,D}(k)).
\end{aligned}
\end{array}
\end{equation}

Then, the VSS~$\mu_i(k)$ for $i=0,...,N-1$ can be derived from the following minimization:
\begin{equation}
\label{050}
\begin{array}{rcl}
\begin{aligned}
\min \limits_{\mu_i(k)} ||e_{p,i}(k)||_2^2,
\end{aligned}
\end{array}
\end{equation}
namely, we have
\begin{equation}
\label{051}
\begin{array}{rcl}
\begin{aligned}
\mu_i(k) = \frac{|e_{i,D}(k)|}{||\bm u_i(k)||_2}.
\end{aligned}
\end{array}
\end{equation}
In the light of the algorithm's stability, values of $\mu_i(k)$ must be bounded as follows:
\begin{equation}
\label{052}
\mu_i(k) = \left\{ \begin{aligned}
&\mu_\text{max}, \text{ if } \mu_i(k) > \mu_\text{max} \\
&\mu_\text{min}, \text{ if } \mu_i(k) < \mu_\text{max}\\
&\mu_i(k), \text{otherwise},
\end{aligned} \right.
\end{equation}
where $\mu_\text{min}$ and $\mu_\text{max}$ denote the lower and upper bounds for $\mu_i(k)$, respectively.
It is worth noting that $\mu_\text{max}$ can be chosen by $\mu_\text{max}=\sqrt{\sigma_d^2/(M\sigma_u^2)}$ to guarantee the stability of the algorithm, where $\sigma_d^2$ and $\sigma_u^2$ denote the powers of $d(n)$ and $u(n)$, respectively, and $\mu_\text{min}$ is selected to be close to zero (e.g., $10^{-5}$).

Furthermore, during the convergence of the algorithm, when the impulsive noise appears, $\mu_i(k)$ will increase immediately to $\mu_\text{max}$, and finally degrading the convergence performance. To avoid this shortcoming, the proposed VSS is revised based on an exponential window strategy as
\begin{equation}
\label{053}
\begin{array}{rcl}
\begin{aligned}
\mu_{o,i}(k) = \beta \mu_{o,i}(k-1)  + (1-\beta) \min\{\mu_i(k),\;\mu_{o,i}(k-1)\},
\end{aligned}
\end{array}
\end{equation}
where the exponential window factor $\beta$ is in general chosen by $\beta=1-N/(\tau M)$ with $\tau \geq 1$, and the initial step-size $\mu_{o,i}(0)$ equals to $\mu_\text{max}$.

\subsection{Adaptation of $\rho$}
To derive the adaptation of~$\rho$, we use~$\rho(k)$ to rewrite~\eqref{eq:041b}~as
\begin{equation}
\label{054}
\begin{array}{rcl}
\begin{aligned}
\widetilde{\bm w}(k+1) = \widetilde{\bm \varphi}(k+1) + \rho(k) H'(\bm \varphi(k+1)).
\end{aligned}
\end{array}
\end{equation}
Taking the squared $l_2$-norm for both sides of~\eqref{054}, we arrive at the following equation
\begin{equation}
\label{055}
\begin{array}{rcl}
\begin{aligned}
||\widetilde{\bm w}(k+1)||_2^2 =& ||\widetilde{\bm \varphi}(k+1)||_2^2 +  \\
&\rho(k) \widetilde{\bm \varphi}^\text{T}(k+1) H'(\bm \varphi(k+1)) +  \\
&\rho^2(k) ||H'(\bm \varphi(k+1))||_2^2.
\end{aligned}
\end{array}
\end{equation}
Letting $||\widetilde{\bm w}(k+1)||_2^2$ be minimum with respect to $\rho(k)$, the optimal $\rho(k)$ is formulated as
\begin{equation}
\label{056}
\begin{array}{rcl}
\begin{aligned}
\rho_o(k) = -\frac{\widetilde{\bm \varphi}^\text{T}(k+1) H'(\bm \varphi(k+1))}{||H'(\bm \varphi(k+1))||_2^2}. \\
\end{aligned}
\end{array}
\end{equation}
Applying~\eqref{0D4} at any iteration~$k$ into~\eqref{056}, we get
\begin{equation}
\label{057}
\begin{array}{rcl}
\begin{aligned}
\rho_o(k) \geq \frac{H(\bm \varphi(k+1)) - H(\bm w^o)}{||H'(\bm \varphi(k+1))||_2^2}. \\
\end{aligned}
\end{array}
\end{equation}
Equation~\eqref{057} is not practical owing mainly to requiring the \emph{a~priori} sparsity $H(\bm w^o)$.
To estimate this sparsity, our previous work~\cite{yu2019sparsity} is employed to yield an estimate of $\bm w^o$:
\begin{equation}
\label{058}
\begin{array}{rcl}
\begin{aligned}
&\text{if}\;k==0\\
&\;\;\;\;\hat{\bm w} = \bm \varphi(k+1)\\
&\text{else} \\
&\;\;\;\;\hat{\bm w} = 0.5\hat{\bm w} + 0.5\bm \varphi(k+1)\\
&\text{end}.
\end{aligned}
\end{array}
\end{equation}
Note that, using~$\hat{\bm w}$ would lead to $H(\bm \varphi(k+1)) - H(\hat{\bm w}) < 0$ so that $\rho_o(k)<0$. As such, based on this consideration, from~\eqref{057} we propose the following rule for adjusting~$\rho_o(k)$:
\begin{equation}
\label{059}
\begin{array}{rcl}
\begin{aligned}
\rho_o(k) = \chi \frac{\max[H(\bm \varphi(k+1)) - H(\hat{\bm w}),\; 0]}{||H'(\bm \varphi(k+1))||_2^2}, \\
\end{aligned}
\end{array}
\end{equation}
where the parameter~$\chi$ results from the inequality sign in~\eqref{057} and based on extensive simulation results, we found out that $1\leq \chi \leq2$ works well.

As a result, by equipping the S-IWF-SSAF algorithm with $\mu_{o,i}(k)$ and $\rho_o(k)$, we refer to it as the variable parameters S-IWF-SSAF (VP-S-IWF-SSAF) algorithm and summarize in Table~\ref{table_2}.
\begin{table}[htbp]
    \scriptsize
    \centering
    \vspace{-1em}
    \caption{Proposed VP-S-IWF-SSAF algorithm.}
    \label{table_2}
    \begin{tabular}{lc}
        \hline
        \text{Initializations:} $\bm w(k) = \bm 0$, $\mu_{o,i}(0)=\mu_\text{max}$;\\
        \text{Parameters:} $1\leq \chi \leq2$;\\
        \;\;\;\;\;\;\;\;\;\;\;\;\;\;\;\;\;\; $\mu_\text{max}=\sqrt{\sigma_d^2/(M\sigma_u^2)}$, large step size; \\
        \;\;\;\;\;\;\;\;\;\;\;\;\;\;\;\;\;\; $\mu_{\min}>0$, very small step size, e.g., $10^{-5}$;\\
        \;\;\;\;\;\;\;\;\;\;\;\;\;\;\;\;\;\; $\delta>0$, very small number to avoid the division by zero;\\
        \;\;\;\;\;\;\;\;\;\;\;\;\;\;\;\;\;\; $\beta=1-N/(\tau M)$, exponential weighted factor, with $\tau \geq 1$;\\
        \;\;\;\;\;\;\;\;\;\;\;\;\;\;\;\;\;\; $\xi>0$, small constant to distinguish active and inactive entries;\\
        \hline
        \text{for} \text {each iteration \emph{k}} \text{do}\\
        \text{ }\text{ }\text{ }\text{ } $e_{i,D}(k) = d_{i,D}(k) - \bm u^\text{T}_i(k) \bm w(k)$,\;$i=0,...,N-1$ \\
        \text{ }\text{ } \text{ } \text{VSS mechanism:} \\
        \text{ }\text{ }\text{ }\text{ } $\mu_i(k) = \frac{|e_{i,D}(k)|}{\|\bm u_i(k)\|_2+10^{-5}}$, \;$i=0,...,N-1$\\
        \text{ }\text{ }\text{ }\text{ } $\mu_i(k) = \left\{ \begin{aligned}
        &\mu_\text{max}, \text{ if } \mu_i(k) > \mu_\text{max} \\
        &\mu_\text{min}, \text{ if } \mu_i(k) < \mu_\text{max}\\
        &\mu_i(k), \text{otherwise}
        \end{aligned} \right.$\\
        \text{ }\text{ }\text{ }\text{ }\text{ }$\mu_{o,i}(k) = \beta \mu_{o,i}(k-1)  + (1-\beta) \min\{\mu_i(k),\;\mu_{o,i}(k-1)\}$ \\
        \text{ }\text{ }\text{ }\text{ }\text{ }$\bm \varphi(k+1) = \bm w(k) + \sum \limits_{i=0}^{N-1} \mu_{o,i}(k) \frac{\text{sgn}(e_{i,D}(k))\bm u_i(k)}{\sqrt{\|\bm u_i(k)\|_2^2+\delta}} $ \\
        \text{ }\text{ } \text{ } \text{Adaptation of $\rho$:} \\
        \text{ }\text{ }\text{ }\text{ }\text{ }$H(\bm \varphi(k+1)) = \sum_{m=1}^{M} \ln(1+|\varphi_m(k+1)|/\xi)$ \\
        \text{ }\text{ }\text{ }\text{ }\text{ }$H'(\varphi_m(k+1)) = \frac{\text{sgn}(\varphi_m(k+1))}{\xi + |\varphi_m(k+1)|}, m=1,...,M$\\
        \text{ }\text{ }\text{ }\text{ }\text{ }$\begin{aligned}
        &\text{if}\; k==0\\
        &\;\;\;\;\rho_o(k) = 0\\
        &\text{else} \\
        &\;\;\;\;\rho_o(k) = \chi \frac{\max[H(\bm \varphi(k+1)) - H(\hat{\bm w}),\; 0]}{||H'(\bm \varphi(k+1))||_2^2}\\
        &\text{end}
        \end{aligned}$\\
        \text{ }\text{ }\text{ }\text{ }\text{ }$\bm w(k+1) = \bm \varphi(k+1)- \rho_o(k) H'(\bm \varphi(k+1))$ \\
        \text{ }\text{ }\text{ }\text{ }\text{ }$\begin{aligned}
        &\text{if}\;k==0\\
        &\;\;\;\;\hat{\bm w} = \bm \varphi(k+1)\\
        &\text{else} \\
        &\;\;\;\;\hat{\bm w} = 0.5\hat{\bm w} + 0.5\bm \varphi(k+1)\\
        &\text{end}
        \end{aligned}$\\
        \text{ }\text{ }\text{ }\text{ }\text{ }$H(\hat{\bm w}) = \sum_{m=1}^{M} \ln(1+|\hat{w}_m|/\xi)$ \\
        \text{end}\\
        \hline
    \end{tabular}
\end{table}
{ \subsection{Discussions} \textit{Remark~3-link to some existing
algorithms:} the proposed S-IWF-SSAF algorithm is obtained by the
two-steps implementation given by~\eqref{eq:b38a}
and~\eqref{eq:b38b}, which differs from the traditional
sparsity-aware
framework~\cite{haddad2014transient,gu2009norm,de2014sparsity,yu2019sparsity}.
This framework makes for the development of the VP-S-IWF-SSAF
algorithm in Section~5. Note that, the S-IWF-SSAF algorithm becomes
the original IWF-SSAF algorithm when the sparsity penalty parameter
$\rho=0$ in~\eqref{eq:b38b}. In the VP-S-IWF-SSAF algorithm, the VSS
mechanism is inspired by the idea to derive the BDVSS-IWF-SSAF
algorithm in~\cite{kim2017delayless}, as it does not require~\emph{a
priori} knowledge of the background noise variance. Moreover, the
proposed VSS removes redundant terms like the one
in~\cite{kim2017delayless}. Roughly speaking, the proposed
VP-S-IWF-SSAF algorithm is a sparsity-aware modification of the
BDVSS-IWF-SSAF algorithm. Importantly, the VP-S-IWF-SSAF algorithm
employs an adaptive rule to choose the sparsity penalty
parameter~$\rho$.

\textit{Remark~4:} In the context of system identification, Table~\ref{table_3x} compares the computational complexity of the NSAF, IWF-SSAF, BDVSS-SSAF and BDVSS-IWF-SSAF algorithms with that of the S-IWF-SSAF and VP-S-IWF-SSAF algorithms, in terms of the total number of additions, multiplications, divisions, and square-roots per iteration~$k$\footnote{These amounts will be reduced by a factor of $1/N$ for each fullband input sample~$n$.}. Note that, $P^{+}=2(L-1)N$ and $P^{*}=2LN$ are the inherent additions and multiplications required by the SAF algorithms, for partitioning the input signal $u(n)$ and the desired signal $d(n)$. To exploit the sparsity of the unknown systems, the proposed S-IWF-SSAF algorithm requires extra $3M$~additions and $M$~divisions per iteration to calculate~\eqref{eq:b38a} and~\eqref{040} than the original IWF-SSAF algorithm. Likewise, in contrast to the existing BDVSS-IWF-SSAF algorithm, the proposed VP-S-IWF-SSAF algorithm requires not only the calculation of~\eqref{eq:b38a} and~\eqref{040} but also the adaptation of $\rho_o(k)$ in Section~5.2, that is, leading to extra $11M$~additions, $3M$~multiplications, $M+1-N$~divisions, and $2M$ logarithms per iteration, to improve the filter's performance in sparse systems as we shall see in simulations.
\begin{table*}[tbp]
    \scriptsize
    \centering
    \caption{Number of of Arithmetical Operations per iteration~$k$.}
    \label{table_3x}
    \begin{tabular}{@{}|l|c|c|c|c|c|}
        \hline
        \text{Algorithms} &\text{Additions} &\text{Multiplications} &\text{Divisions} &\text{Square-roots} &\text{Logarithms}\\
        \hline
        \text {NSAF}  &$(2M+2)N+P^{+}$ &$(2M+3)N+P^{*}$ &$N$ &- &-\\ \hline
        \text {IWF-SSAF}  &$(2M+3)N+P^{+}$ &$(2M+2)N+P^{*}$ &$N$ &$N$ &-\\ \hline
        \text {BDVSS-IWF-SSAF} &$(2M+9)N+P^{+}$ &$(2M+4)N+P^{*}$ &$3N$ &$N$ &-\\ \hline
        \text {S-IWF-SSAF} &$(2M+3)N+3M+P^{+}$ &$(2M+2)N+P^{*}$ &$N+M$ &$N$ &-\\ \hline
        \text {VP-S-IWF-SSAF} &$(2M+9)N+11M+P^{+}$ &$(2M+4)N+3M+P^{*}$ &$2N+M+1$ &$N$ &$2M$\\
        \hline
    \end{tabular}
\end{table*}
}

\section{Simulation Results}
In this section, extensive simulations are presented to verify our theoretical analysis and proposed algorithms. It is assumed that the length of the adaptive filter matches that of the unknown system. We use the cosine modulated analysis filter banks with $N$ subbands in the SAF structure. All of the results are the ensemble average of 200 independent trials, unless otherwise specified.

\subsection{Theoretical verifications}
In the system identification, the elements in $\bm w^o$ to be estimated are randomly generated according to the uniform distribution $[-0.5, 0.5]$. The used input signal~$u(n)$ originates from a first-order autoregressive (AR) system, $u(n)=0.9u(n-1)+\epsilon(n)$, where $\epsilon(n)$ is a zero-mean white Gaussian signal with variance~$\sigma_\epsilon^2=1$ except in Fig.~\ref{Fig2}(d). Such AR input has a high correlation relative to the white input with $u(n)=\epsilon(n)$. The additive noise $v(n)$ at the unknown system's output follows a CG random process, described in~assumption~2. In the CG, the background noise component~$v_\text{g}(n)$ with variance $\sigma_\text{g}^2$ gives rise to a signal-to-noise ratio (SNR) defined as $10\log_{10}(\sigma_{\bar{d}}^2/\sigma_\text{g}^2)$, where $\sigma_{\bar{d}}^2 = \text{E}\{(\bm u^\text{T}(n) \bm w^o)^2\}$ is the output signal power of the unknown system in noise-free environments. The parameter $\hbar$ for the impulsive noise component~$v_\text{im}(n)$ is set to $\hbar = 300000$, which determines the impulsive characteristic for its realizations. The expectations $\text{E}\{\bm A_i(k)\}, \text{E}\{\check{\bm A}_i(k)\}, \bm R_i$, and $\sigma_{u,i}^2$ in the theoretical models are obtained by the available ensemble average.

Figs.~\ref{Fig2} and~\ref{Fig3} show the steady-state MSDs as a function of $\mu$ from 0.01 to 0.35, where the steady-state MSDs are obtained by averaging 500 instantaneous MSD values in the steady-state. The theoretical stability upper bound is computed according to~\eqref{040x1} with the relation~$(b)$ in~\eqref{040x2}. In Fig.~\ref{Fig2}, we investigate the effect of $p_r$, $N$, SNR, and $\sigma_\epsilon^2$, respectively, by varying one of them, on the stability condition for the IWF-SSAF algorithm. In Fig.~\ref{Fig2}, to investigate the effect of $M$ on the algorithm's stability, we normalize  $\bm w^o$ as $||\bm w^o||_2=1$. As can be seen from~Figs.~\ref{Fig2} and~\ref{Fig3}, under the impulsive noise environment with small occurrence probability~$p_r$, the theoretical stability range is valid for guiding the choice of~$\mu$. Moreover, large~$p_r$ and small SNR make the stability upper bound shrink. Interestingly, the input power $\sigma_\epsilon^2$ and $N$ do not seem to affect the stability bound of the algorithm.
\begin{figure*}[htb]
    \centering
    \includegraphics[scale=0.45] {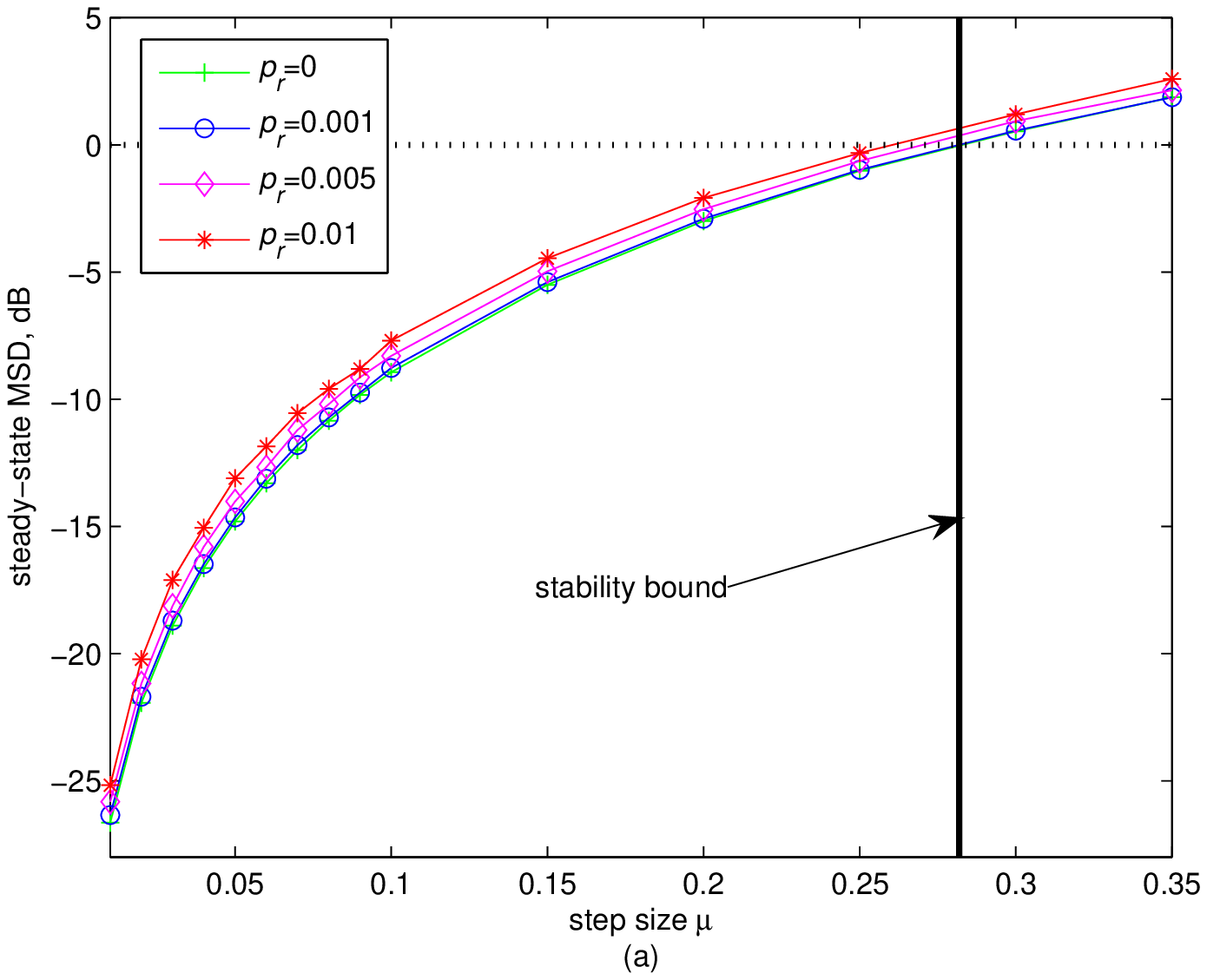}
    \hspace{10ex}
    \includegraphics[scale=0.45] {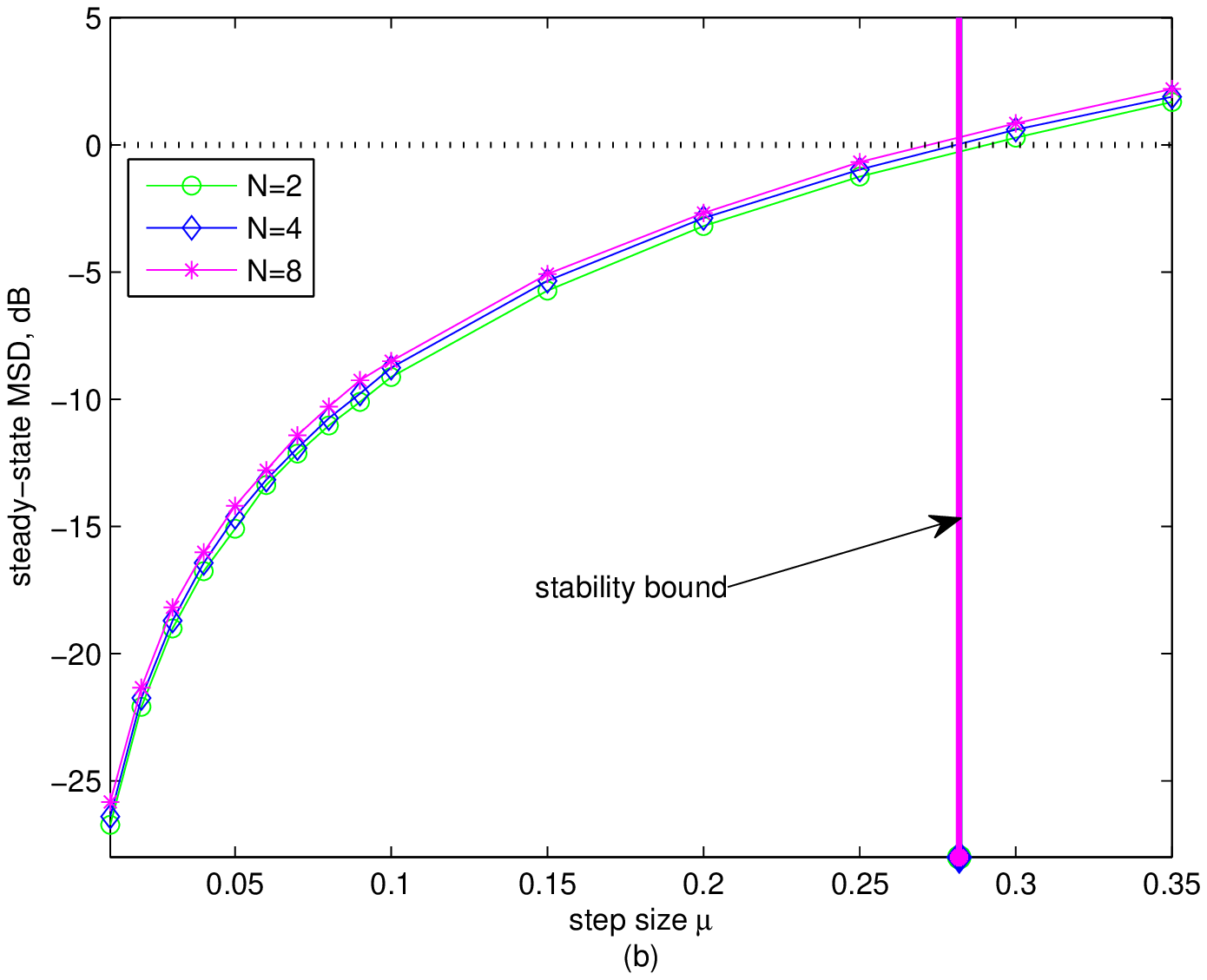}

    \includegraphics[scale=0.45] {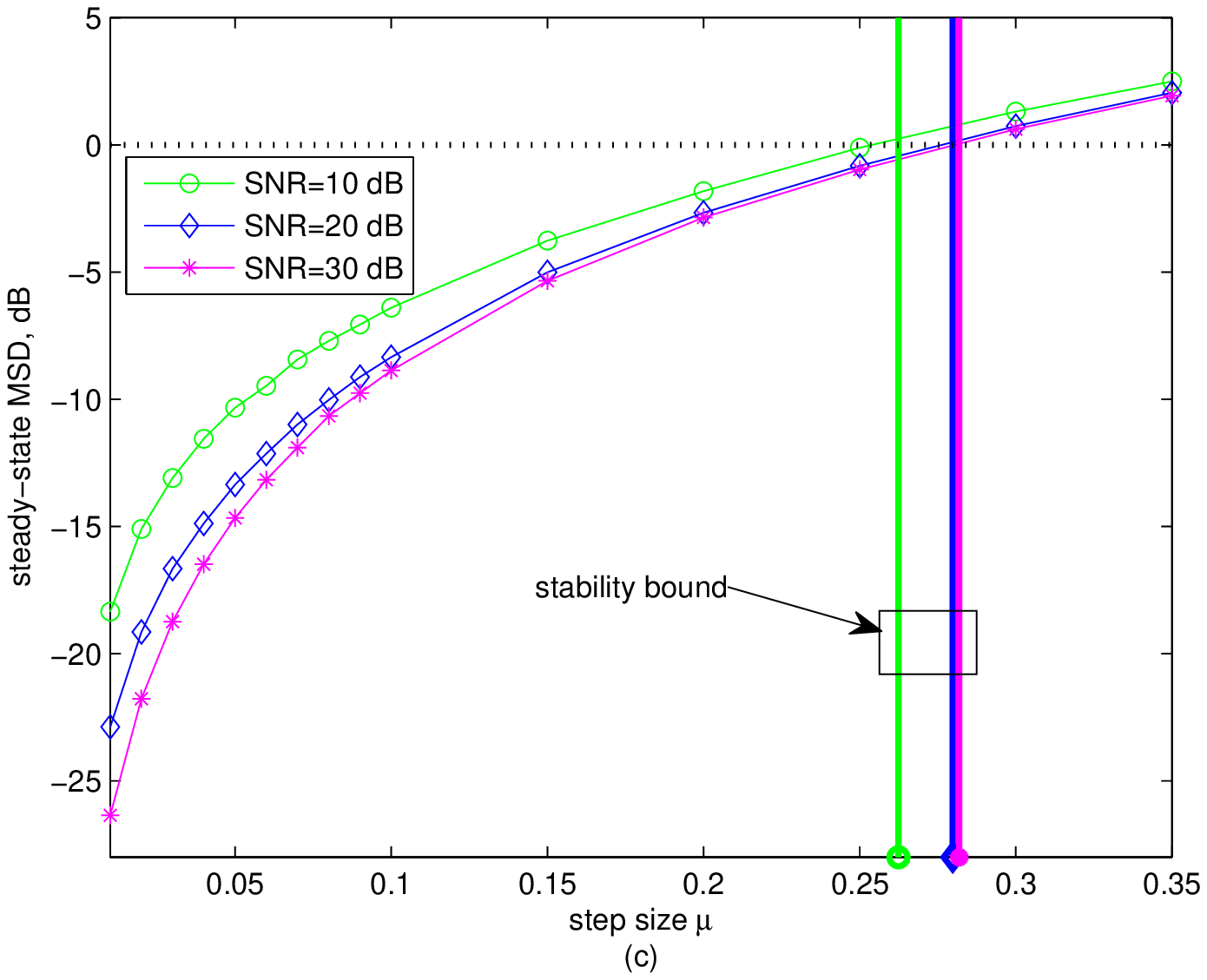}
    \hspace{10ex}
    \includegraphics[scale=0.45] {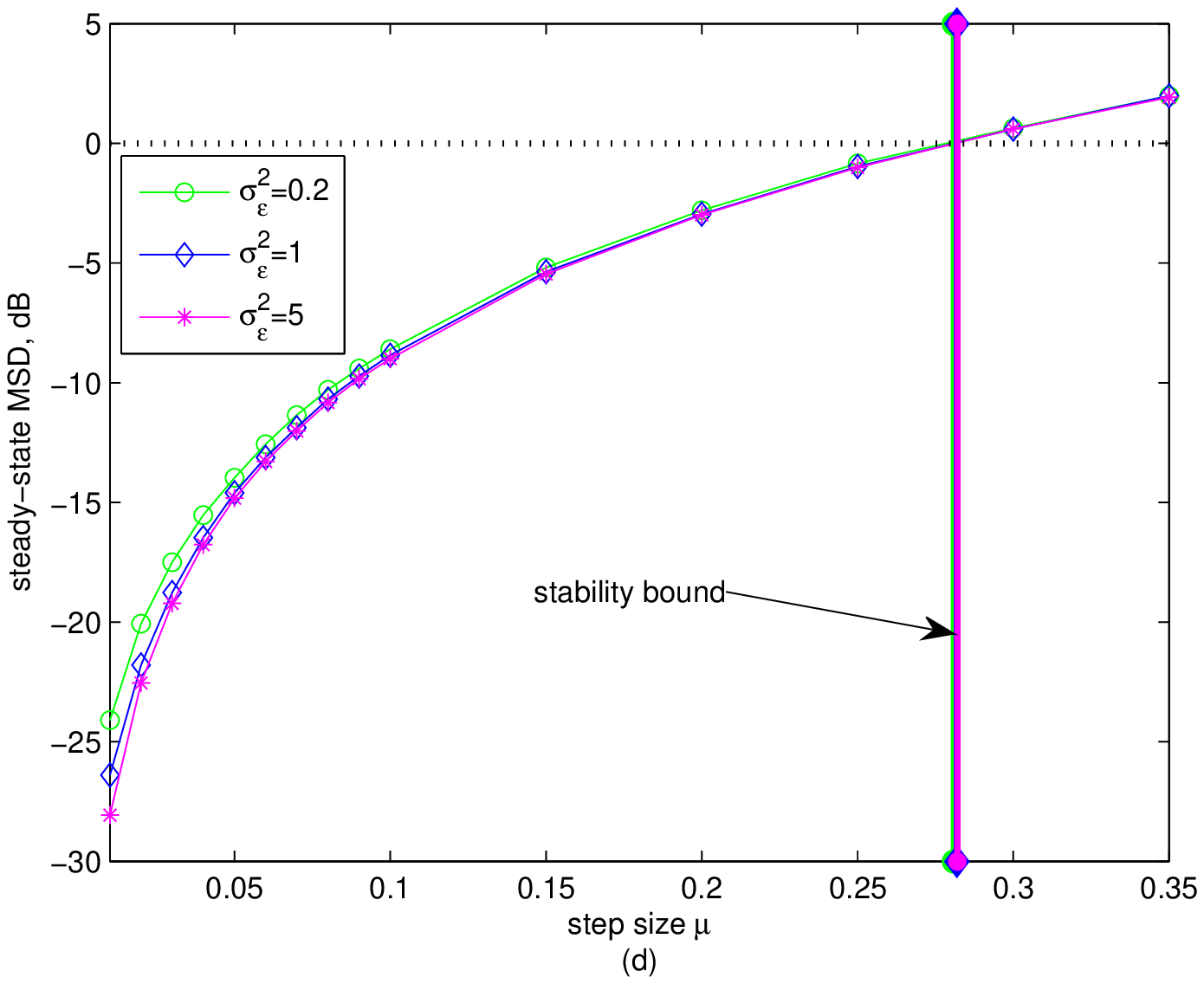}
    \vspace{-1em} \caption{Steady-state MSDs versus step size~$\mu$. (a) $N=4$, $\sigma_\epsilon^2=1$, SNR=30~dB; (b) $p_r=0.001$, $\sigma_\epsilon^2=1$, SNR=30~dB; (c) $N=4$, $\sigma_\epsilon^2=1$, $p_r=0.001$; (d) $N=4$, $p_r=0.001$, SNR=30~dB. [$M=32$].}
    \label{Fig2}
\end{figure*}
\begin{figure}[htb]
    \centering
    \includegraphics[scale=0.45] {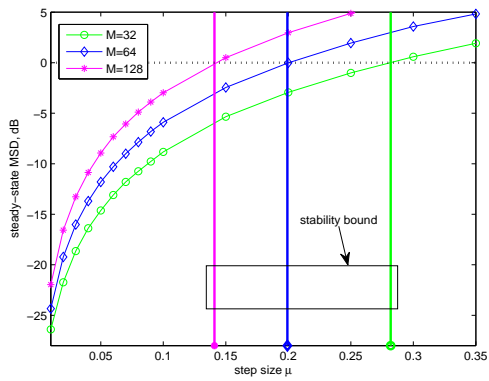}
    \vspace{-1em} \caption{Steady-state MSDs versus step size~$\mu$. [$p_r=0.001$, $N=4$, $\sigma_\epsilon^2=1$, SNR=30~dB].}
    \label{Fig3}
\end{figure}

Figs.~\ref{Fig4} to~\ref{Fig6} check the proposed transient model~\eqref{eq:043} (where $\rho=0$) for the IWF-SSAF algorithm and compares with the transient model in~\cite{yu2016novel}. As can be seen, the proposed model matches better the simulation than the existing model. This is due mainly to the fact that the existing model relies on the long adaptive filter assumption and the whitening assumption on the decimated subband input vectors, while the proposed model does not require them. For instance, in Fig.~\ref{Fig4}, the steady-state results for the existing model are closer to the simulations as~$M$ increases. It is seen from~Fig.~\ref{Fig5} that the fixed step-size~$\mu$ can not make the IWF-SSAF algorithm reach fast convergence and low steady-state MSD simultaneously. By increasing the number of subbands~$N$, the IWF-SSAF algorithm can speed up the convergence, see Fig.~\ref{Fig6}. Moreover, Figs.~\ref{Fig5} and~\ref{Fig6} reveal that the theory expression~\eqref{046} with~\eqref{028} on the steady-state MSD of the IWF-SSAF algorithm is valid only when the step-size is small.
\begin{figure}[htb]
    \centering
    \includegraphics[scale=0.45] {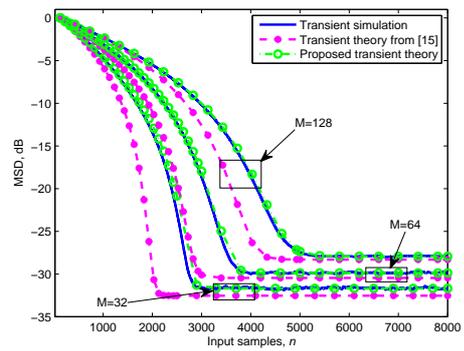}
    \vspace{-1em} \caption{MSD curves of the IWF-SSAF algorithm for different~$M$. [$p_r=0.001$, $\mu=0.004$, $N=4$, SNR=30~dB].}
    \label{Fig4}
\end{figure}
\begin{figure}[htb]
    \centering
    \includegraphics[scale=0.45] {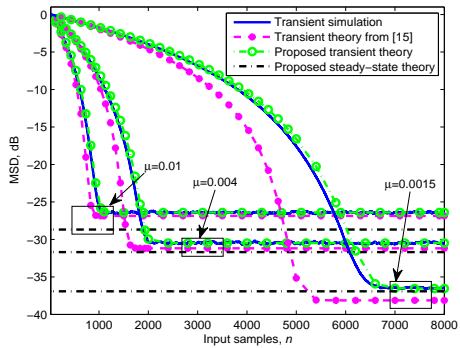}
    \vspace{-1em} \caption{MSD curves of the IWF-SSAF algorithm using different step sizes. [$p_r=0.001$, $N=4$, $M=32$, SNR=30~dB].}
    \label{Fig5}
\end{figure}
\begin{figure}[htb]
    \centering
    \includegraphics[scale=0.45] {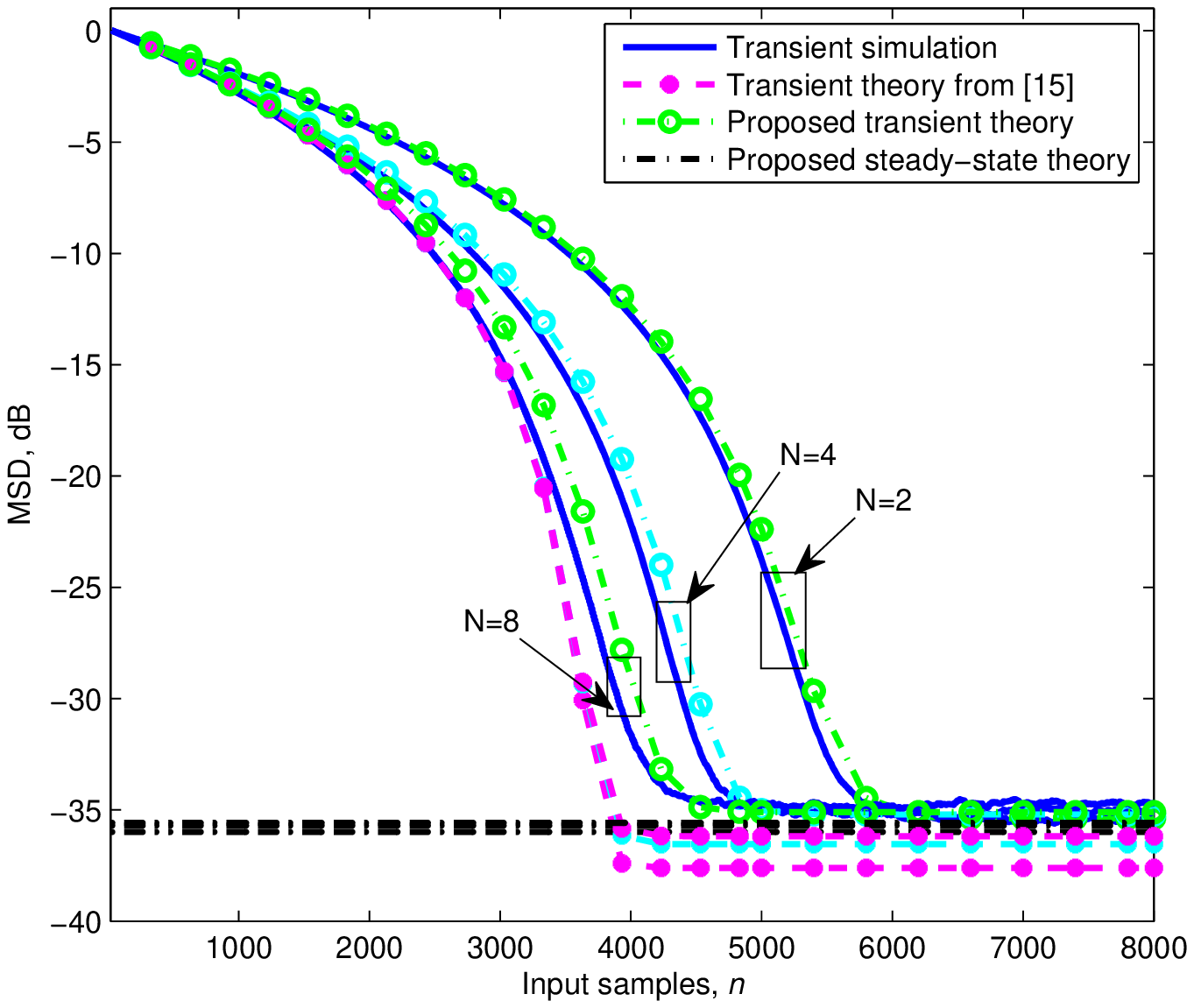}
    \vspace{-1em} \caption{MSD curves of the IWF-SSAF algorithm for different~$N$. [$p_r=0.001$, $\mu=0.002$, $M=32$, SNR=30~dB].}
    \label{Fig6}
\end{figure}

In Figs.~\ref{Fig8} and~\ref{Fig9}, we check the theoretical insights given in Section~IV.~B on the S-IWF-SSAF algorithm in sparse scenarios. The sparse vector $\bm w^o$ of interest has $M=64$ entries, where its NZ entries are Gaussian variables with zero mean and variance of $1/\sqrt{|\text{NZ}|}$ and their positions are randomly selected from the binomial distribution. As one can see in Fig.~\ref{Fig8}, in terms of the transient MSD of the S-IWF-SSAF algorithm, the theoretical curves obtained from~\eqref{eq:043} have good fit with the simulated curves. In order to show the better performance of the S-IWF-SSAF algorithm when the parameter vector of interest is sparse (where $|\text{NZ}|=4$) as compared to the IWF-SSAF algorithm, we choose $\rho=4\times10^{-5}$ in Fig.~\ref{Fig8}(a) and $\rho=7\times10^{-5}$ in Fig.~\ref{Fig8}(b), respectively. Figs.~\ref{Fig9} investigates the effect of $\rho$ on the S-IWF-SSAF performance in the steady-state. As declared in theorem~2 or~\eqref{0D2}, in sparse systems there is a range of values for choosing~$\rho$ so that the S-IWF-SSAF algorithm outperforms the IWF-SSAF algorithm. Moreover, this range will gradually die out as $\bm w^o$ becomes non-sparse (i.e., NZ entries become more).
\begin{figure}[htb]
    \centering
    \includegraphics[scale=0.45] {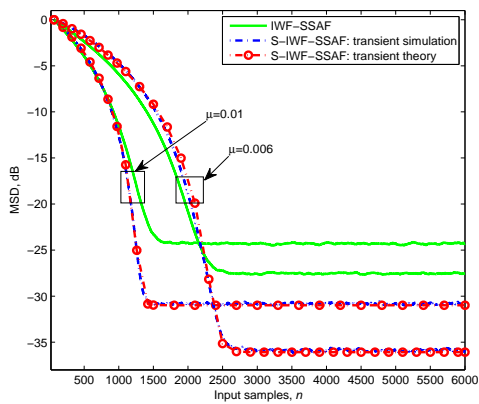}
    \vspace{-1em} \caption{MSD curves of the IWF-SSAF and S-IWF-SSAF algorithms for identifying sparse~$\bm w^o$. [$\xi=0.05$, $p_r=0.001$, $N=4$, SNR=30~dB].}
    \label{Fig8}
\end{figure}
\begin{figure}[htb]
    \centering
    \includegraphics[scale=0.46] {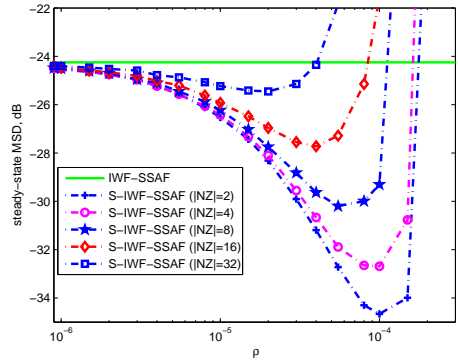}
    \vspace{-1em} \caption{Steady-state MSDs versus~$\rho$. [$\xi=0.05$, $p_r=0.001$, $\mu=0.01$, $N=4$, SNR=30~dB].}
    \label{Fig9}
\end{figure}
\subsection{Comparison of algorithms}
In this subsection, we compare the performance of the proposed S-IWF-SSAF and VP-S-IWF-SSAF algorithms with that of the NSAF, BDVSS-SSAF~\cite{yoo2014band}, IWF-SSAF, and BDVSS-IWF-SSAF~\cite{kim2017delayless} algorithms under the AEC environment. In the hands-free telephone system, the echo is frequently encountered, that is, the talker hears his own time-delayed voice~\cite{hansler2005acoustic}. Concretely, $\bm w^o$ denotes the acoustic echo channel between loudspeaker and microphone at the near-end. The far-end speech $u(n)$ is played at the loudspeaker, and after passing through $\bm w^o$ to yield the echo signal $y(n) = \bm u^T(n) \bm w^o$; at the same time, the echo signal $y(n)$ is picked up by the microphone and sent to the far-end talker, which impair the quality of the speech. In adaptive AEC, by feeding the same input signal $u(n)$, the output of the adaptive filter $\bm w(n)$ will be the replica of the echo, i.e., $\hat{y}(n) = \bm u^T(n) \bm w(n)$, thus by performing $e(n) = d(n)-\hat{y}(n)$ (where the microphone signal~$d(n)$ consists of the echo, the background noise $v_g(n)$, and the possible near-end speech or impulsive noise~$v_\text{im}(n)$), we can cancel the echo as $\bm w(n) \rightarrow \bm w^o$. That is to say, the adaptive AEC is a typical adaptive system identification problem that identifies the acoustic echo channel, even if we eventually need the signal~$e(n)$. To address the delay issue from the original SAF structure depicted in Fig.~\ref{Fig1} for an AEC application, however, Fig.~\ref{FigAEC} shows its delayless structure~\cite{lee2007delayless}. In comparison, the only difference in the delayless structure is that $e(n)$ is calculated in the original sequences through an auxiliary loop by copying $\bm w(k)$ to $\bm w(n)$ once for every $N$ input samples (i.e., when $n=kN$). Here, the sparse acoustic echo channel is shown in Fig.~\ref{Fig10} with $M=512$ taps. We set the number of subbands to~$N=8$ for all the SAF algorithms.
\begin{figure}[htb]
    \centering
    \includegraphics[scale=0.4] {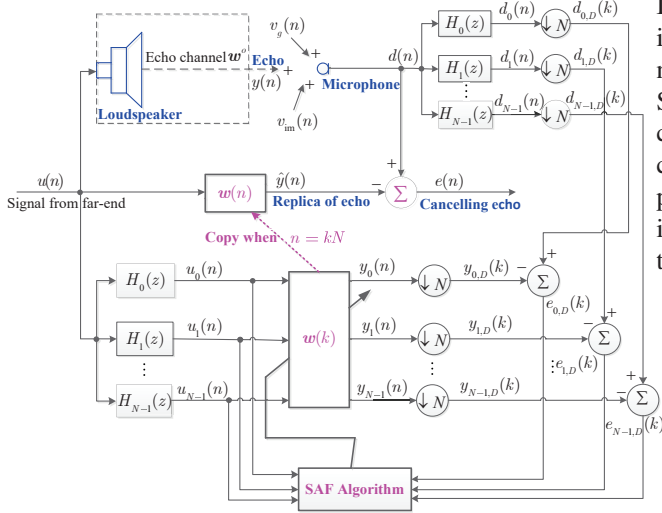}
    \vspace{-1em} \caption{Delayless multiband structure of SAF applied into AEC.}
    \label{FigAEC}
\end{figure}
\begin{figure}[htb]
    \centering
    \includegraphics[scale=0.45] {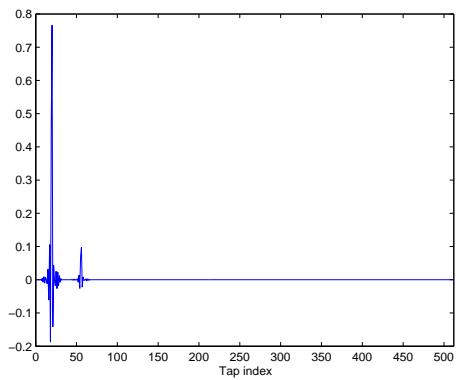}
    \vspace{-1em} \caption{Sparse acoustic echo channel.}
    \label{Fig10}
\end{figure}

In the first example, we consider the (symmetric) $\alpha$-stable process to describe the additive noise $v(n)$ with impulsive samples~\cite{nikias1995signal,georgiou1999alpha,pelekanakis2014adaptive}, which is more realistic than the CG process usually used in the analysis. The $\alpha$-stable noise is expressed by $\phi(t) = \exp(-\gamma |t|^\alpha)$, where the characteristic exponent $0<\alpha \leq 2$ determines the noise's impulsiveness (whose role behaves like the impulsive noise probability), and $\gamma>0$ indicates the dispersion level of the noise. Note that, when $\alpha=2$, it is the Gaussian noise. In view of acoustic scenarios, we set $\alpha=1.5$ and $\gamma=1/30$~\cite{georgiou1999alpha}. The normalized MSD, i.e., $\text{NMSD}(n) = 10\log10(\text{MSD}(n)/||\bm w^o||_2^2)$ is used for measuring the algorithms' performance. Fig.~\ref{Fig11} depicts the NMSD performance of the algorithms for the AR input. In this figure, to compare the tracking capability of these algorithms, $\bm w^o$ undergoes a sudden change by shifting its~12 taps to the right at the 80001-th input sample. As expected, all the SSAF-type algorithms show stable convergence in the $\alpha$-stable noise, while the NSAF algorithm has severe shaking. Since the IWF-SSAF algorithm employs the band-dependent weighting factors rather than the common one in the SSAF algorithm, the BVDSS-IWF-SSAF algorithm has better convergence than the BDVSS-SSAF algorithm. Due to the sparsity-aware step~\eqref{eq:b38b}, the proposed S-IWF-SSAF algorithm significantly enhance the IWF-SSAF's steady-state performance when identifying sparse channels. Moreover, by proposing the adaptation of $\rho$, the S-IWF-SSAF with $\rho_o(k)$ algorithm overcomes the selection problem of $\rho$ in the S-IWF-SSAF algorithm. Benefited from the adaptively adjusting parameters ($\mu$ and $\rho$), the proposed VP-S-IWF-SSAF algorithm is superior to the other algorithms in terms of convergence and steady-state behaviors. It is noticed that like the BDVSS-SSAF and BDVSS-IWF-SSAF algorithms, the VP-S-IWF-SSAF algorithm also cannot track the sudden change of $\bm w^o$, but this issue can be addressed by applying the reset algorithm (RA) presented in~\cite{yoo2014band}. By using the speech signal as the input, similar results can be obtained in Fig.~\ref{Fig12} except the tracking performance.
\begin{figure}[htb]
    \centering
    \includegraphics[scale=0.45] {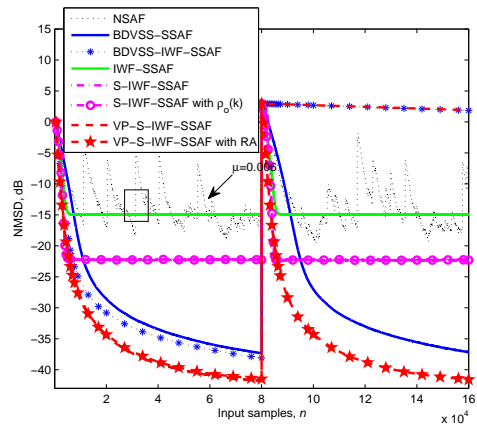}
    \vspace{-1em} \caption{NMSD curves of various SAF algorithms for the AR input. We set the fixed step-size $\mu=0.2$ for the NSAF algorithm and $\mu=0.01$ for the IWF-SSAF, S-IWF-SSAF, and S-IWF-SSAF with $\rho_o(k)$ algorithms. The parameter~$\rho$ for the S-IWF-SSAF algorithm is chosen as $5\times10^{-6}$ in a trial and error way. For the proposed sparsity-aware algorithms, we choose $\xi=0.01$, $\chi=2$ for the S-IWF-SSAF with $\rho_o(k)$ and $\chi=1$ for the VP-S-IWF-SSAF. Other parameters of algorithms are set as follows: $V_T=M$, $\kappa=1$ for the BDVSS-SSAF; $\mu_U=0.06$, $\mu_L=10^{-5}$, $\kappa=1$ for the BDVSS-IWF-SSAF; $\tau=1$, $\mu_{\min}=10^{-5}$ for the VP-S-IWF-SSAF.}
    \label{Fig11}
\end{figure}
\begin{figure}[htb]
    \centering
    \includegraphics[scale=0.45] {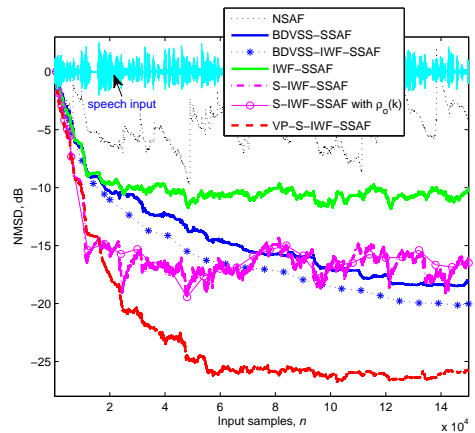}
    \vspace{-1em} \caption{NMSD curves of various SAF algorithms for the speech input (single run). Since the speech signal exists the silent period (zero values), to avoid the division by zero in the algorithms' update, we add a regularization parameter~$\delta$, i.e., $\delta=20\sigma_u^2/N$ (NSAF and proposed algorithms), $\delta=20\sigma_u^2$ (BDVSS-SSAF and BDVSS-IWF-SSAF). Some algorithms' parameters are retuned as $\kappa=2$ (BDVSS-SSAF), $\kappa=3$ (BDVSS-IWF-SSAF), and $\tau=2$ (VP-S-IWF-SSAF).}
    \label{Fig12}
\end{figure}

In the second example, we consider the double-talk case that the near-end speech also appears. In this case, we choose $\alpha=1.8$ for the $\alpha$-stable noise. The NMSD results of the algorithms are shown in Fig.~\ref{Fig12}. As opposed to the NSAF algorithm, the SSAF-type algorithms (especially for the VSS versions) are insensitive to the double-talk. The VP-S-IWF-SSAF algorithm is still the best choice among these algorithms, because it optimizes the selections of the step-size and the sparsity penalty parameter.
\begin{figure}[htb]
    \centering
    \includegraphics[scale=0.45] {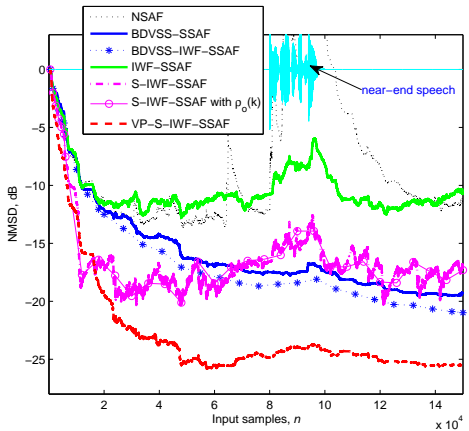}
    \vspace{-1em} \caption{NMSD curves of various SAF algorithms in the double-talk scenario (single run). Parameters setting is the same as in Fig.~\ref{Fig12}.}
    \label{Fig13}
\end{figure}
\section{Conclusion}
In this work, the S-IWF-SSAF algorithm was proposed to take advantage of the underlying sparsity of the systems. The theoretical analysis of the S-IWF-SSAF algorithm has been performed that it has significant improvement in the steady-state performance as compared to the IWF-SSAF counterpart when works in sparse system environments. Even for the IWF-SSAF algorithm, the proposed analysis does not require special assumptions, so it matches more accurate with the simulated results than the existing analysis. Moreover, we have developed joint time-varying schemes of both the step-size and the sparsity penalty parameter for the S-IWF-SSAF algorithm, i.e., VP-S-IWF-SSAF, to further improve the convergence and steady-state performance. Simulations in both system identification and AEC situations have been conducted to verify our theoretical analysis and the effectiveness of the proposed algorithms.

{  It is worth pointing out that, for the proposed VP-S-IWF-SSAF
algorithm we simply choose the log-penalty given in~\eqref{0b39}
from several sparsity-aware strategies as a paradigm. Therefore, in
the future, it is also necessary to study the effect of different
sparsity-aware strategies with respect to the effectiveness of joint
variable parameters in this algorithm.}

\section*{Acknowledgments}
This work was partially supported by the National Natural Science Foundation of China (NSFC) (Nos. 61901400), and the Doctoral Research Fund of Southwest University of Science and Technology in China (No. 19zx7122).

\appendix
\renewcommand{\appendixname}{Appendix~}
\section{Derivation of~\eqref{eq:042a}}

Taking the expectation for both sides of~\eqref{eq:041a}, we have
\begin{equation}
\label{0A1}
\begin{aligned}
\text{E}\{\widetilde{\bm \varphi}(k+1)\} = \text{E}\{\widetilde{\bm w}(k)\} - \mu \sum \limits_{i=0}^{N-1} \text{E} \left\lbrace  \frac{\text{sgn}(e_{i,D}(k))\bm u_i(k)}{\|\bm u_i(k)\|_2} \right\rbrace.
\end{aligned}
\end{equation}
To compute the last expectation in~\eqref{0A1}, we introduce the conditional expectation:
\begin{equation}
\label{0A2}
\begin{aligned}
&\text{E} \left\lbrace  \frac{\text{sgn}(e_{i,D}(k))\bm u_i(k)}{\|\bm u_i(k)\|_2} \right\rbrace = \\
&\;\;\;\;\;\;\;\;\;\;\;\;\;\;\;\;\;\;\text{E} \left\lbrace \text{E} \left\lbrace  \frac{\text{sgn}(e_{i,D}(k))\bm u_i(k)}{\|\bm u_i(k)\|_2} \left|\widetilde{\bm w}(k)\right.   \right\rbrace \right\rbrace.
\end{aligned}
\end{equation}

Since $\bm h_i$ is deterministic, it follows~\eqref{006} and assumption~2 that $v_{i,D}(k)$ can also be referred to as a CG process, i.e., $v_{i,D}(k)=v_{\text{g},i,D}(k) + b_i(k)\eta_{i,D}(k)$, where $v_{\text{g},i,D}(k)$ and $\eta_{i,D}(k)$ are zero-mean white Gaussian with variances $\sigma_{\text{g},i}^2 = ||\bm h_i||_2^2 \sigma_\text{g}^2$ and $\sigma_{\eta,i}^2=||\bm h_i||_2^2 \sigma_\eta^2$ respectively, and $b_i(k)$ obeys the Bernoulli distribution with $P\{b_i(k)=1\}=p_r$ being the probability of occurring 1. It should be stressed that if the analysis filter bank for partitioning the input signal $u(n)$ and the desired signal $d(n)$ is assumed to be identical and paraunitary, as done in references~\cite{yin2011stochastic} for analyzing the NSAF algorithm, then $||\bm h_i||_2^2=1/N$ will be further given.

Hence, applying the CG noise model and  the law of total probability, the following relation is established:
\begin{equation}
\label{0A3}
\begin{aligned}
&\text{E} \left\lbrace \frac{\text{sgn}(e_{i,D}(k))\bm u_i(k)}{\|\bm u_i(k)\|_2} \left|\widetilde{\bm w}(k) \right. \right\rbrace = \\
&P\{b(k)=1\} \text{E} \left\lbrace \frac{\text{sgn}(e_{i,D,1}(k))\bm u_i(k)}{\|\bm u_i(k)\|_2} \left|\widetilde{\bm w}(k) \right. \right\rbrace + \\
&P\{b(k)=0\}\text{E} \left\lbrace \frac{\text{sgn}(e_{i,D,2}(k))\bm u_i(k)}{\|\bm u_i(k)\|_2} \left|\widetilde{\bm w}(k) \right. \right\rbrace = \\
&\;\;\;\;\;\;\;\;\;\;\;\;\;\;p_r \text{E} \left\lbrace \frac{\text{sgn}(e_{i,D,1}(k))\bm u_i(k)}{\|\bm u_i(k)\|_2} \left|\widetilde{\bm w}(k) \right. \right\rbrace + \\
&\;\;\;\;\;\;\;\;\;\;\;\;\;\;(1-p_r)\text{E} \left\lbrace \frac{\text{sgn}(e_{i,D,2}(k))\bm u_i(k)}{\|\bm u_i(k)\|_2} \left|\widetilde{\bm w}(k) \right. \right\rbrace,
\end{aligned}
\end{equation}
where $e_{i,D,1}(k) = e_{i,a}(k)+v_{\text{g},i,D}(k)+\eta_{i,D}(k)$ and $e_{i,D,2}(k) = e_{i,a}(k)+v_{\text{g},i,D}(k)$.

Both $v_{\text{g},i,D}(k)$ and $\eta_{i,D}(k)$ are Gaussian random variables, thus we can assume $e_{i,D,1}(k)$ and $e_{i,D,2}(k)$ to be zero mean Gaussian variables when~$M$ is large~\cite{al2003transient}. Accordingly, based on Price's theorem\footnote{If real-valued random variables $x$ and $y$ are jointly Gaussian, it holds that $\text{E}\{x\text{sgn}(y)\} = \sqrt{2/\pi}\text{E}\{xy\} / \sqrt{\text{E}\{y^2\}}$~\cite{price1958useful}.} and assumption~3, we have
\begin{equation}
\label{0A4}
\begin{aligned}
&\text{E} \left\lbrace \frac{\text{sgn}(e_{i,D,1}(k))\bm u_i(k)}{\|\bm u_i(k)\|_2} \left|\widetilde{\bm w}(k) \right. \right\rbrace= \\
&\;\;\;\;\;\;\;\;\;\;\;\sqrt{\frac{2}{\pi}} \frac{1}{\sqrt{\text{E} \{e_{i,D,1}^2(k)\}}} \text{E}\{\bm A_i(k)\} \widetilde{\bm w}(k),
\end{aligned}
\end{equation}
and
\begin{equation}
\label{0A5}
\begin{aligned}
&\text{E} \left\lbrace \frac{\text{sgn}(e_{i,D,2}(k))\bm u_i(k)}{\|\bm u_i(k)\|_2} \left|\widetilde{\bm w}(k) \right. \right\rbrace= \\
&\;\;\;\;\;\;\;\;\;\;\;\sqrt{\frac{2}{\pi}} \frac{1}{\sqrt{\text{E} \{e_{i,D,2}^2(k)\}}} \text{E}\{\bm A_i(k)\} \widetilde{\bm w}(k),
\end{aligned}
\end{equation}
where $\bm A_i(k) = \frac{\bm u_i(k) \bm u_i^\text{T}(k)}{\|\bm u_i(k)\|_2}$.

By substituting \eqref{0A2}-\eqref{0A5} into~\eqref{0A1}, we will obtain ~\eqref{eq:042a}.
\section{Derivation of~\eqref{eq:043a}}
Both sides of \eqref{eq:041a} are multiplied by their transposes, then we take the expectations of all the terms to yield
\begin{equation}
\label{0B1}
\begin{aligned}
\widetilde{\bm \varPhi}&(k+1)  = \widetilde{\bm W}(k) - \\
&\mu  \sum \limits_{i=0}^{N-1} \underbrace{\text{E}\left\lbrace \frac{\widetilde{\bm w}(k) \text{sgn}(e_{i,D}(k)) \bm u_i^\text{T}(k)}{\|\bm u_i(k)\|_2} \right\rbrace } \limits_\text{I} - \\
&\mu \sum \limits_{i=0}^{N-1} \underbrace{\text{E} \left\lbrace \frac{\bm u_i(k) \text{sgn}(e_{i,D}(k)) \widetilde{\bm w}^\text{T}(k) }{\|\bm u_i(k)\|_2}\right\rbrace }\limits_\text{II} + \\
&\mu^2 \underbrace{ \text{E}\left\lbrace  \sum \limits_{i=0}^{N-1} \frac{\bm u_i(k) \text{sgn}(e_{i,D}(k))}{\|\bm u_i(k)\|_2} \sum \limits_{j=0}^{N-1} \frac{\bm u_j^\text{T}(k) \text{sgn}(e_{j,D}(k))}{\|\bm u_j(k)\|_2} \right\rbrace }\limits_\text{III},
\end{aligned}
\end{equation}

By performing similar procedures to those in~\eqref{0A2}-\eqref{0A5}, the term I in~\eqref{0B1} can be calculated by
\begin{equation}
\label{0B2}
\begin{aligned}
\text{E}\left\lbrace \frac{\widetilde{\bm w}(k) \text{sgn}(e_{i,D}(k)) \bm u_i^\text{T}(k)}{\|\bm u_i(k)\|_2} \right\rbrace = \widetilde{\bm W}(k) \Omega_i(k) \text{E}\{\bm A_i(k)\},
\end{aligned}
\end{equation}
but we omit this derivation for brevity. The term II in~\eqref{0B1} is the transpose of~\eqref{0B2}, i.e.,
\begin{equation}
\label{0B3}
\begin{aligned}
\text{E} \left\lbrace \frac{\bm u_i(k) \text{sgn}(e_{i,D}(k)) \widetilde{\bm w}^\text{T}(k) }{\|\bm u_i(k)\|_2}\right\rbrace = \Omega_i(k) \text{E}\{\bm A_i(k)\} \widetilde{\bm W}(k) .
\end{aligned}
\end{equation}

The term~III in~\eqref{0B1} can be approximated as $\sum \limits_{i=0}^{N-1} \text{E}\{\check{\bm A}_i(k)\}$, where $\check{\bm A}_i(k) = \frac{\bm u_i(k) \bm u_i^\text{T}(k)}{\|\bm u_i(k)\|_2^2}$, because different subband vectors $\bm u_i(k)$ and $\bm u_j(k)$ are weakly correlated~\cite{lee2009subband,yin2011stochastic}. With these relations, it is easy to derive~\eqref{eq:043a} from~\eqref{0B1}.

\section{}
For ease of evaluating~$\text{E}\{H'(\bm \varphi(k))\}$, $\bm \varTheta(k)$, and $\bm \varXi(k)$, we need another two assumptions. They have been frequently called up for simplifying the analyses of sparsity-aware adaptive filtering algorithms~\cite{haddad2014transient,yu2019sparsity,TSP2012gu}.

\textit{Assumption 4}: The $m$-th component of the intermediate weights error vector $\widetilde{\bm \varphi}(k)$ for every~$k$, has a Gaussian distribution, namely, $\widetilde{\varphi}_m(k) \sim \aleph(z_m(k),\sigma_m^2(k))$, where the mean $z_m(k)$ is the $m$-th component of $\text{E}\{\widetilde{\bm \varphi}(k)\}$ from~\eqref{eq:042a} and the variance $\sigma_{m}^2(k)$ is computed from~\eqref{eq:043a} by $\sigma_m^2(k) = \widetilde{\bm \varPhi}_{m,m}(k)-z_m^2(k)$. As such, $\varphi_m(k)$ follows the distribution~$\aleph(\bar{\varphi}_m,\sigma_{\varphi,m}^2)$ with $\bar{\varphi}_m=w_m^o-z_m(k)$ and $\sigma_{\varphi,m}^2=\sigma_m^2(k)$ where $w_m^o$ is the $m$-th component of $\bm w^o$.

\textit{Assumption 5}: When $m\neq l$, it can be assumed that $\text{E}\{H'(\varphi_m(k)) \varphi_l(k)\} \approx \text{E}\{H'(\varphi_m(k))\} \text{E}\{ \varphi_l(k)\}$ and $ \text{E}\{H'(\varphi_m(k)) H'(\varphi_l(k))\} \approx \text{E}\{H'(\varphi_m(k))\} \text{E}\{H'(\varphi_l(k))\}$.

Based on assumption~4, we compute the $m$-th component of $\text{E}\{H'(\bm \varphi(k)\}$ by
\begin{equation}
\label{0C1}
\begin{aligned}
\text{E}\{H'(\varphi_m(k))\} \approx \frac{\text{E}\{\text{sgn}(\varphi_m(k))\}}{\xi + \text{E}\{|\varphi_m(k)|\}},
\end{aligned}
\end{equation}
where this approximation also because $\xi$ is relatively small, and
\begin{equation}
\label{0C2}
\begin{array}{rcl}
\begin{aligned}
\text{E}\{|\varphi|\} =& \frac{1}{\sqrt{2\pi}\sigma_\varphi} \int_{-\infty}^{\infty}|\varphi|\exp^{-\left( \frac{\varphi-\bar{\varphi}}{\sqrt{2}\sigma_\varphi}\right)^2}d\varphi \\
=&\sqrt{\frac{2}{\pi}} \sigma_\varphi \exp^{-\frac{\bar{\varphi}^2}{2\sigma_\varphi^2}} + \bar{\varphi} \text{erf}\left( \frac{\bar{\varphi}}{\sqrt{2}\sigma_\varphi}\right),
\end{aligned}
\end{array}
\end{equation}
and
\begin{equation}
\label{0C3}
\begin{array}{rcl}
\begin{aligned}
\text{E}\{\text{sgn}(\varphi)\} =& \frac{1}{\sqrt{2\pi}\sigma_\varphi} \int_{-\infty}^{\infty}\text{sgn}(\varphi)\exp^{-\left( \frac{\varphi-\bar{\varphi}}{\sqrt{2}\sigma_\varphi}\right)^2}d\varphi \\
=&\text{erf}\left(\frac{\bar{\varphi}}{\sqrt{2}\sigma_\varphi}\right)
\end{aligned}
\end{array}
\end{equation}
with $\text{erf}(\varphi)\triangleq \frac{2}{\sqrt{\pi}}\int_{0}^{\varphi}\exp^{-t^2}dt$. It has been given that $\bm \varTheta_{m,l}(k) \triangleq \text{E}\{H'( \varphi_m(k))\}{w_l^o} - \text{E}\{H'(\varphi_m(k)) \varphi_l(k)\}$ and $\bm \varXi_{m,l}(k) \triangleq \text{E}\{H'(\varphi_m(k)) H'(\varphi_l(k))\}$ for any $m,l$. Specifically, when $m=l$, we have
\begin{equation}
\label{0C4}
\begin{aligned}
\text{E}\{H'(\varphi_m(k)) \varphi_m(k)\} \approx \frac{\text{E}\{|\varphi_m(k)|\}}{\xi + \text{E}\{|\varphi_m(k)|\}}
\end{aligned}
\end{equation}
and
\begin{equation}
\label{0C5}
\begin{aligned}
\text{E}\{H'(\varphi_m(k))^2\} \approx \frac{1}{\xi^2 + 2\xi\text{E}\{|\varphi_m(k)|\} + \text{E}\{\varphi_m^2(k)\}}.
\end{aligned}
\end{equation}
When $m\neq l$, according to assumption~5 we can obtain
\begin{equation}
\label{0C6}
\begin{aligned}
\text{E}\{H'(\varphi_m(k)) \varphi_l(k)\} \approx \frac{\text{E}\{\text{sgn}(\varphi_m(k))\}}{\xi + \text{E}\{|\varphi_m(k)|\}} \text{E}\{\varphi_l(k)\}
\end{aligned}
\end{equation}
and
\begin{equation}
\label{0C7}
\begin{aligned}
\text{E}\{H'(\varphi_m(k)) H'(\varphi_l(k))\} \approx \frac{\text{E}\{\text{sgn}(\varphi_m(k))\}}{\xi + \text{E}\{|\varphi_m(k)|\}} \frac{\text{E}\{\text{sgn}(\varphi_l(k))\}}{\xi + \text{E}\{|\varphi_l(k)|\}}.
\end{aligned}
\end{equation}

\section{Possibility of $\Delta_s(\infty)<0$}
The matrix $(\bm I_{M^2} - \bm F_\infty)$ is positive definite due to the convergence of the algorithm. As such, it can be approximated by $(\bm I_{M^2} - \bm F_\infty)\approx t \cdot \bm I_{M^2}$~\cite{golub2012matrix}, where $t>0$ is finite. In this case, we are able to simplify~$\Delta_s(\infty)$ given in~\eqref{046x} as
\begin{equation}
\label{0D1}
\begin{array}{rcl}
\begin{aligned}
\Delta_s&(\infty) \approx  \frac{1}{t} \text{vec}^\text{T}(\bm I_M) \text{vec} \left( \rho \bm \varTheta(\infty) + \rho \bm \varTheta^\text{T}(\infty) + \rho^2 \bm \varXi(\infty) \right) \\
=& \frac{1}{t} \left[ 2\rho\text{E}\{\widetilde{\bm \varphi}^\text{T}(\infty) H'(\bm \varphi(\infty)) \} + \rho^2 \text{E}\{||H'(\bm \varphi(\infty))||_2^2\} \right], \\
=& \frac{\rho}{t} \text{E}\{||H'(\bm \varphi(\infty))||_2^2\} \left[ \rho + \frac{2\text{E}\{\widetilde{\bm \varphi}^\text{T}(\infty) H'(\bm \varphi(\infty)) \}}{\text{E}\{||H'(\bm \varphi(\infty))||_2^2\}} \right] .
\end{aligned}
\end{array}
\end{equation}
Obviously, $\Delta_s(\infty)<0$ is true if and only if
\begin{equation}
\label{0D2}
\begin{array}{rcl}
\begin{aligned}
0<\rho < -\frac{2\text{E}\{\widetilde{\bm \varphi}^\text{T}(\infty) H'(\bm \varphi(\infty)) \}}{\text{E}\{||H'(\bm \varphi(\infty))||_2^2\}} \triangleq \rho_\text{up}
\end{aligned}
\end{array}
\end{equation}
and
\begin{equation}
\label{0D3}
\begin{array}{rcl}
\begin{aligned}
\text{E}\{\widetilde{\bm \varphi}^\text{T}(\infty) H'(\bm \varphi(\infty)) \} < 0.
\end{aligned}
\end{array}
\end{equation}
The above relations shows that~$\text{E}\{\widetilde{\bm \varphi}^\text{T}(\infty) H'(\bm \varphi(\infty)) \}<0$ is a necessary condition to $\Delta_s(\infty)<0$ in the scenario that the vector $\bm w^o$ to be estimated is sparse. In particular, the number of zero elements (Z set) in $\bm w^o$ are much more than that of non-zero elements (NZ set).

Since $H(\cdot)$ is a real-valued convex function, from the definition of the sub-gradient~\cite{yu2019sparsity,spar2013}, the following inequality will hold:
\begin{equation}
\label{0D4}
\begin{array}{rcl}
\begin{aligned}
\widetilde{\bm \varphi}^\text{T}(\infty) H'(\bm \varphi(\infty)) &= (\bm w^o - \bm \varphi(\infty))^\text{T} H'(\bm \varphi(\infty)) \\
&\leq  H(\bm w^o) -  H(\bm \varphi(\infty)).
\end{aligned}
\end{array}
\end{equation}

Since the penalty function~$H(\bm w)$ in~\eqref{0b39} measures the sparsity of the vector $\bm w$ and undoubtedly, the true vector~$\bm w^o$ is more sparse than its estimate~$\bm \varphi(\infty)$ obtaining from the gradient update~\eqref{eq:b38a}, it is expected that
\begin{equation}
\label{0D5}
\begin{array}{rcl}
\begin{aligned}
H(\bm w^o) -  H(\bm \varphi(\infty)) &\approx - \sum_{m\in \text{Z set}}^{M} \ln(1+|\varphi_m(\infty)|/\xi) \\
& < 0.
\end{aligned}
\end{array}
\end{equation}
Consequently, by appropriately choosing the sparse penalty parameter~$\rho$, the condition $\Delta_s(\infty)<0$ is likely to be true. Conversely, when $\bm w^o$ is not sparse,~\eqref{0D5} would not hold so that $\Delta_s(\infty)<0$ is impossible regardless of $\rho$.


\bibliography{mybibfile}

\end{document}